\documentclass[10pt,journal,twocolumn,twoside]{IEEEtran}

\usepackage{theorem} \usepackage{cite} \usepackage{stfloats} \usepackage{epsfig} \usepackage{verbatim}
\usepackage{graphicx} \usepackage{amssymb} \usepackage{amsmath} \usepackage{color}
\usepackage{bm}
\usepackage{xcolor}
\usepackage{algorithm}
\usepackage{algorithmic}
\usepackage{makecell}
\usepackage{subfigure}

\ifCLASSOPTIONcompsoc
	\usepackage[caption=false,font=normalsize,labelfont=sf,textfont=sf]{subfig}
\else
	\usepackage[caption=false,font=footnotesize]{subfig}
\fi

\theoremstyle{plain}
\newtheorem{theorem}{Theorem}
\theoremstyle{plain}
\newtheorem{proposition}{Proposition}
\theoremstyle{plain}
\newtheorem{corollary}{Corollary}
\theoremstyle{plain}
\newtheorem{lemma}{Lemma}
\theoremstyle{plain}

\newcommand\va{{\bf a}} 

\newcommand\vc{{\bf c}}
\newcommand\vd{{\bf d}}

\newcommand\vf{{\bf f}}

\newcommand\vh{{\bf h}}

\newcommand\vr{{\bf r}}
\newcommand\vs{{\bf s}}

\newcommand\vu{{\bf u}}
\newcommand\vv{{\bf v}}
\newcommand\vw{{\bf w}}
\newcommand\vx{{\bf x}}

\newcommand\mA{{\bf A}} 

\newcommand\mE{{\bf E}}
\newcommand\mF{{\bf F}}
\newcommand\mG{{\bf G}}
\newcommand\mH{{\bf H}}

\newcommand\mR{{\bf R}}
\newcommand\mS{{\bf S}}

\newcommand\mU{{\bf U}}

\newcommand\mW{{\bf W}}
\newcommand\mX{{\bf X}}

\newcommand\defi{{\triangleq}}

\renewcommand{\baselinestretch}{1}

\begin{document}

\def\QEDclosed{\mbox{\rule[0pt]{1.3ex}{1.3ex}}}
\def\QEDopen{{\setlength{\fboxsep}{0pt}\setlength{\fboxrule}{0.2pt}\fbox{\rule[0pt]{0pt}{1.3ex}\rule[0pt]{1.3ex}{0pt}}}}
\def\QED{\QEDopen}
\def\proof{}
\def\endproof{\hspace*{\fill}~\QED\par\endtrivlist\unskip}

\title{A Partial Channel Reciprocity-based Codebook for Wideband FDD Massive MIMO}

\author{Haifan~Yin
    and David~Gesbert,~\IEEEmembership{Fellow,~IEEE}
\thanks{

This work was supported by National Natural Science Foundation of China under grant No. 62071191. The corresponding author is Haifan Yin.
}
\thanks{H. Yin is with Huazhong University of Science and Technology, 430074 Wuhan, China (e-mail: yin@hust.edu.cn).
}%
\thanks{D. Gesbert is with EURECOM, 06410 Biot, France (e-mail: gesbert@eurecom.fr).}
}

\maketitle

\begin{abstract}

The acquisition of channel state information (CSI) in Frequency Division Duplex (FDD) massive MIMO has been a formidable challenge. In this paper, we address this problem with a novel CSI feedback framework enabled by the partial reciprocity of uplink and downlink channels in the wideband regime. We first derive the closed-form expression of the rank of the wideband massive MIMO channel covariance matrix for a given angle-delay distribution. A low-rankness property is identified, which generalizes the well-known result of the narrow-band uniform linear array setting. Then we propose a partial channel reciprocity (PCR) codebook, inspired by the low-rankness behavior and the fact that the uplink and downlink channels have similar angle-delay distributions. Compared to the latest codebook in 5G, the proposed PCR codebook scheme achieves higher performance, lower complexity at the user side,  and requires a smaller amount of feedback. We derive the feedback overhead necessary to achieve asymptotically error-free CSI feedback. Two low-complexity alternatives are also proposed to further reduce the complexity at the base station side. Simulations with the practical 3GPP channel model show the significant gains over the latest 5G codebook, which prove that our proposed methods are practical solutions for 5G and beyond.

\end{abstract}


\newcounter{EquCounter}

\section{Introduction}\label{sec_intro}

The large-scale commercialization of 5G cellular systems gradually brings the concept of massive multiple-input multiple-output (MIMO) \cite{marzetta:10a} to reality. With a large number of antennas at the base station, massive MIMO leads to much higher spectral and energy efficiencies compared with traditional MIMO.
The stronger pair-wise orthogonality between the channel vectors of different User Equipments (UEs) allows for higher interference rejection capability for the base stations. In other words, the narrow beam of the massive antenna array reduces the leakage of the transmitted signal towards unwanted UEs, which is one of the main features of massive MIMO.
Nevertheless, the performance of massive MIMO is primarily determined by the accuracy of the CSI. In practical cellular systems, there are mainly three reasons of inaccurate CSI, which are pilot contamination \cite{jose2011TWC}, the mobility problem  \cite{yin2020JSAC}, and the CSI feedback with limited overhead in FDD \cite{marzetta:10a}.

The uplink (UL) and downlink (DL) in FDD occupy different frequency bands, which are separated by around 100 MHz in current 5G setups - a gap much larger than the channel coherence bandwidth. Therefore, unlike in Time Division Duplex (TDD) mode, the UL and DL channels in FDD are typically non-reciprocal. As a result, the DL CSI is mainly obtained by closed-loop feedback methods: The base station (BS) first transmits reference signals, e.g., Channel State Information Reference Signal (CSI-RS) as in 5G, then the UEs estimate their individual channels, and finally the UEs send back the quantized channel information to the base station. This traditional CSI feedback scheme suffers from two problems, 1) the time-frequency resource spent on the reference signals increases quickly with the number of base station antennas, leaving less resource for data transmission; and 2) the CSI feedback is always corrupted by quantization errors.

In the literature, there have been many proposals to enhance the performance of FDD massive MIMO. The works in \cite{yin:13} and \cite{adhikary:13} indicated that the spatial channel covariance matrix of the UE channels exhibits a certain low-rankness property, where the rank is governed by the angular spread of the multipath. Such a property helps to reject interference from non-orthogonal training signals of other UEs which have non-overlapping angular distribution with the desired UE, thus reducing the training overhead. In some other related works, an idea of covariance shaping was proposed and studied \cite{newinger2015} \cite{moghadam2017} \cite{mursia2018} \cite{khalilsarai2018}, where the effective channels with low-dimensionality were manually created and exploited through precoding.
Some papers proposed to adopt Deep Learning for CSI compression and reconstruction \cite{Wen2018, Liang2020, Lu2020}. However these learning-based method generally require training with large datasets.

The possibility of exploiting the reciprocity in angular domain has been studied in \cite{luo2017, ding2018, zhang2018, shen2018}. In particular, \cite{luo2017} proposed a scalable CSI feedback method which exploited the angle and delay reciprocity of the FDD channel. A Fast Fourier Transform (FFT)-based pilots were used and the UE computed its low-dimensional feedback based on the DL channel statistics measured by itself or instructed by its serving base station. The authors found that only a scalar per each DL channel path was needed for the base station to reconstruct the DL channel.
\cite{ding2018} proposed a dictionary learning approach which relied on the reciprocity and the sparsity of the multipath angles in UL and DL channels. \cite{zhang2018} proposed a directional training method which was also motivated by the angular reciprocity. The authors in \cite{shen2018} utilized the angular reciprocity in designing an angle-adaptive codebook and proved that the amount of feedback scales with the number of paths, which was assumed much smaller than the number of base station antennas.

In Release-16 (Rel-16) of the 5G standards, the Enhanced Type II Codebook is adopted \cite{3gpp:38.214}. Such a codebook capitalizes on the sparsity in both spatial domain and frequency domain by introducing a 2D Discrete Fourier Transform (DFT) operation on the wideband channel matrix. 
Briefly speaking, the base station first broadcasts non-precoded CSI-RS to its serving UEs.
The UE then estimates its wideband channel matrix.
Afterwards, the UE performs 2D DFT on the wideband precoding matrix, where each column of this matrix is a precoder for a certain frequency band.
Finally the UE sends back the non-negligible elements (quantized) in the transformed matrix to the base station along with their corresponding positions in the transformed matrix.

In this paper, we propose a novel CSI feedback framework based on an idea of joint spatial-frequency domain precoded reference signal. This framework is enabled by the partial UL/DL reciprocity of angle-delay distributions, which was recently validated by measurements \cite{Zhong2020}. In the traditional methods, the precoding is always performed in spatial domain. However in this paper, the frequency-domain precoding is made possible by our scheme,  which entails a joint operation of the BS and the UE. The basic idea is to exploit the channel sparsity in angle-delay domain by finding the non-negligible spatial-frequency-domain projections from the wideband uplink channel estimation, and then design the downlink joint precoder based on the positions of these projected values. In this way, we only sample the non-negligible coefficients of the sparse representation of the channel, which leads to significant training overhead reduction. In practice, the projection is realized by either the eigenvectors of the wideband channel covariance matrix or the 2D DFT matrix.

We show that the amount of feedback coefficients can be greatly reduced, which is always smaller than the total number of multipaths. In the meantime, the computational complexity for the UE is quite low - only several summations of the estimates of the effective channels will suffice. This is well-aligned with the commercial interest to simplify the implementation of UEs.
Different from \cite{luo2017}, the UEs do not need any information about the DL channel statistics or path delays, nor do they need to perform path aligning process, and the complexity of our method is much lower.

Compared with the latest Rel-16 codebook, i.e., the Enhanced Type II Codebook, our proposed PCR codebook scheme has five main advantages.
 \begin{itemize}
 \item PCR scheme achieves better performance in terms of the CSI feedback accuracy. It is because of the exploitation of the second-order channel statistics.
 Compared to the DFT-based methods, e.g., the Rel-16 codebook, the dimension of the signal subspace of a channel covariance matrix is smaller than the number of non-negligible DFT projections. This is because of the spectrum leakage of DFT projections, particularly when the size of the DFT matrix is relatively small.

 \item The PCR method has lower computational complexity. In fact the Enhanced Type II Codebook entails a 2D DFT operation, and in some cases, the singular value decomposition (SVD). While our method only requires some scalar additions instead of matrix operations.

 \item The proposed method has a smaller feedback overhead. There are two reasons of the feedback reduction. 1) The eigenvector based compression is more effective than the DFT based compression; and 2) In Rel-16 codebook, apart from feeding back the quantized complex coefficients, the UE has to feed back the indices of the corresponding DFT vectors as well. While in our proposed feedback scheme, the UE only needs to feedback the quantized coefficients.

 \item The proposed method has a better quality of channel estimation at the UE side. In the Enhanced Type II Codebook, the BS broadcasts non-precoded CSI-RS to all UEs. No beamforming is performed at the base station. In such cases, the Signal to Noise Ratio (SNR) of the received CSI-RS might be low.
 In contrast, our proposed codebook scheme enables precoded CSI-RS. In other words, the base station transmits beamformed reference signals to UEs, which is similar to the dominant eigen-beamforming. As a result, the receive SNR of the reference signal is higher.

 \item The proposed scheme relies on no assumption of the topology of the BS antenna array. In fact, the Enhanced Type II Codebook in Rel-16 makes the explicit assumption that the base station is equipped with a Uniform Planar Array (UPA). When UE performs 2D DFT transformations, the size of the DFT matrix is determined by the numbers of columns and rows of antenna elements on the antenna panel. A mismatch will result in obvious performance degradations. However, our propose method works with arbitrary antenna topologies, and is therefore more general.

 \end{itemize}

The contributions of our paper are summarized below:
\begin{itemize}
\item We generalize the known low-rankness property of channel covariance matrix for a narrowband large-scale uniform linear array (ULA) to the narrowband UPA array, and to the wideband UPA array. Closed-form expressions of the ranks are derived.

\item We propose a practical and scalable PCR codebook scheme, which is a very competitive candidate for the standards of 5G and beyond. This scheme exploits the UL and DL channel reciprocity of the statistical signal subspace in joint spatial-frequency domains, as well as the proved low-rankness property of the wideband channel covariance matrix.

\item We make performance analysis and give the upper-bound of the feedback overhead to achieve asymptotically error-free CSI feedback. We show that the feedback overhead is smaller than the number of paths.

\item Two low-complexity alternatives named PCR-E and PCR-D are proposed, which further reduce the complexity at the base station side at the cost of mild performance loss.

\end{itemize}

We evaluate the performances of our proposed method under the realistic channel model of 3GPP, which is widely adopted in industry and the de facto model used for 5G standardization.
Simulation results show that in the wideband massive MIMO regime, the UE only needs to send back tens of scalar values to the base station to achieve near optimal performance, even in the presence of hundreds of multipath.

The organization of this paper is as follows: Sec. \ref{modeling} introduces the channel model of 3GPP and the channel reciprocity modelling for FDD \cite{3gpp:36.897}. Sec.  \ref{sec:covaRank} derives the rank of the channel covariance matrices for a given angle-delay distribution. Sec. \ref{sec:PCR} shows the proposed PCR codebook scheme and the asymptotic performance analysis. Sec. \ref{sec:lowComplex} describes the low-complexity alternatives.
The simulation results are shown in Sec. \ref{sec:numericalResult}. Finally, the conclusions are drawn in Sec. \ref{sec:conclusion}.

Notations: We use boldface to denote matrices and vectors. Specifically, ${\mathbf{I}}$ denotes the identity matrix. ${({\mathbf{X}})^T}$, ${({\mathbf{X}})^*}$, and ${({\mathbf{X}})^H}$ denote the transpose, conjugate, and conjugate transpose of a matrix ${\mathbf{X}}$ respectively. ${({\mathbf{X}})^\dag}$ is the Moore-Penrose pseudoinverse of ${\mathbf{X}}$. $\operatorname{tr}\left\{ \cdot \right\}$ denotes the trace of a square matrix.
${\left\| \cdot \right\|_2}$ denotes the $\ell^2$ norm of a vector.
${\left\| \cdot \right\|_F}$ stands for the Frobenius norm.
$\mathbb{E}\left\{ \cdot \right\}$ denotes the expectation.
${\bf{X}}\otimes{\bf{Y}}$ is the Kronecker product of ${\bf{X}}$ and ${\bf{Y}}$. ${\bf{X}} \odot {\bf{Y}}$ denotes the Hadamard product of ${\bf{X}}$ and ${\bf{Y}}$.
 $\operatorname{vec} (\mX)$ is the vectorization of the matrix $\mX$, and  $\operatorname{unvec} (\vx)$ is the corresponding inverse operation.
${\mathop{\rm diag}\nolimits} \{ {{\va_1,...,\va_N}}\}$ denotes a diagonal matrix or a block diagonal matrix with ${\va_1,...,\va_N}$ at the main diagonal. $\triangleq$ is used for definition. $\mathbb{N}$ and $\mathbb{N}^+$ are the set of non-negative and positive integers respectively.

\section{Channel Models}\label{modeling}
We consider an arbitrary UE in a certain cell. The base station is equipped with a uniform planar array (UPA) of cross-polarized antenna elements. The total number of multipaths is $M$. For a certain path $m$, we denote the Zenith angle Of Arrival (ZOA) as $\theta _\text{m,ZOA}$, the Azimuth angle Of Arrival (AOA) as $\varphi _\text{m,AOA}$, the Zenith angle Of Departure (ZOD) as $\theta _\text{m,ZOD}$, and the Azimuth angle Of Departure (AOD) as $\varphi _\text{m,AOD}$.
The channel impulse response between the $s$-th BS antenna and the $u$-th UE antenna at time $t$ is shown in Eq. (\ref{Eq:channelModel}).

\begin{figure*}[!t]
\normalsize
\setcounter{EquCounter}{\value{equation}}
\begin{align}\label{Eq:channelModel}
\begin{array}{l}
H_{u,s,m}(t) = {\left[ {\begin{array}{*{20}{c}}
{{F_{\text{rx,}u,\theta }}\left( {{\theta _{m\text{,ZOA}}},{\varphi _{m\text{,AOA}}}} \right)}\\
{{F_{\text{rx,}u,\varphi }}\left( {{\theta _{m\text{,ZOA}}},{\varphi _{m\text{,AOA}}}} \right)}
\end{array}} \right]^T}\left[ {\begin{array}{*{20}{c}}
{\exp \left( {j\Phi _{m}^{\theta \theta }} \right)}&{\sqrt {{\kappa^{-1}_{m}}} \exp \left( {j\Phi _{m}^{\theta \varphi }} \right)}\\
{\sqrt {{\kappa _{m}}^{ - 1}} \exp \left( {j\Phi _{m}^{\varphi \theta }} \right)}&{\exp \left( {j\Phi _{m}^{\varphi \varphi }} \right)}
\end{array}} \right]\\
\left[ {\begin{array}{*{20}{c}}
{{F_{\text{tx,s},\theta }}\left( {{\theta _{m\text{,ZOD}}},{\varphi _{m\text{,AOD}}}} \right)}\\
{{F_{\text{tx,s},\varphi }}\left( {{\theta _{m\text{,ZOD}}},{\varphi _{m\text{,AOD}}}} \right)}
\end{array}} \right]\exp \left( {j2\pi \frac{{\hat \vr_{\text{rx,}m}^T{{\bar \vd}_{\text{rx,}u}}}}{{{\lambda _0}}}} \right)\exp \left( {j2\pi \frac{{\hat \vr_{\text{tx,}m}^T{{\bar \vd}_{\text{tx,}s}}}}{{{\lambda _0}}}} \right)\exp \left( {j2\pi \frac{{\hat \vr_{\text{rx,}m}^T\bar \vv}}{{{\lambda _0}}}t} \right)
\end{array}
\end{align}
\hrulefill
\vspace*{4pt}
\end{figure*}
${{F_{{\text{rx,}u},\theta }}\left( {{\theta _{m\text{,ZOA}}},{\varphi _{m\text{,AOA}}}} \right)}$, ${{F_{{\text{rx,}u},\varphi }}\left( {{\theta _{m\text{,ZOA}}},{\varphi _{m\text{,AOA}}}} \right)}$, ${{F_{{\text{tx,}s},\theta }}\left( {{\theta _{m\text{,ZOD}}},{\varphi _{m\text{,AOD}}}} \right)}$, and ${{F_{{\text{tx,}s},\varphi }}\left( {{\theta _{m\text{,ZOD}}},{\varphi _{m\text{,AOD}}}} \right)}$ are the field patterns defined in Section 7.3 of \cite{3gpp:38.901}. $\kappa_m$ is the cross polarization ratio (XPR), which follows a log-Normal distribution. $\left\{ {\Phi _{m}^{\theta \theta },\Phi _{m}^{\theta \varphi },\Phi _{m}^{\varphi \theta },\Phi _{m}^{\varphi \varphi }} \right\}$ are the random initial phases for four different polarization combinations. These random phases are uniformly distributed within $(-\pi, \pi)$. $\lambda_0$ is the wavelength of the center frequency. $\hat{\vr}_{\text{rx},m}$ is the spherical unit vector with Azimuth arrival angle $\varphi_{m, \text{AOA}}$ and Zenith arrival angle $\theta_{m, \text{ZOA}}$:
\begin{equation}\label{Eq:rhat_rx_m}
\hat{\vr}_{\text{rx},m} \triangleq \left[ {\begin{array}{*{20}{c}}
{\sin {\theta_{m, \text{ZOA}}}\cos {\varphi_{m, \text{AOA}}}}\\
{\sin {\theta_{m, \text{ZOA}}}\sin {\varphi_{m, \text{AOA}}}}\\
{\cos {\theta_{m, \text{ZOA}}}}
\end{array}} \right].
\end{equation}
Likewise, $\hat{\vr}_{\text{tx},m}$ is the spherical unit vector defined as:
\begin{equation}\label{Eq:rhat_tx_m}
\hat{\vr}_{\text{tx},m} \triangleq \left[ {\begin{array}{*{20}{c}}
{\sin {\theta_{m, \text{ZOD}}}\cos {\varphi_{m, \text{AOD}}}}\\
{\sin {\theta_{m, \text{ZOD}}}\sin {\varphi_{m, \text{AOD}}}}\\
{\cos {\theta_{m, \text{ZOD}}}}
\end{array}} \right].
\end{equation}
$\bar{\vd}_{\text{rx}, u}$ is the 3D Cartesian coordinate of the $u$-th UE antenna, and $\bar{\vd}_{\text{tx}, s}$ is the 3D Cartesian coordinate of the $s$-th base station antenna.
The exponential term $e^{j({\hat{\vr}_{\text{rx},p}^T \bar{\vv} }/{\lambda_0}) t}$ is the Doppler, where $\bar{\vv}$ is the UE velocity vector with speed $v$, travel azimuth angle $\varphi_v$, and travel Zenith angle $\theta_v$:
\begin{equation}\label{Eq:v_bar}
\bar{\vv} \triangleq v {[\begin{array}{*{20}{c}}
{\sin {\theta _v}\cos {\varphi _v}}&{\sin {\theta _v}\sin {\varphi _v}}&{\cos {\theta _v}}
\end{array}]^T}.
\end{equation}
Channel reciprocity is modeled according to \cite{3gpp:36.897} and \cite{kyosti2007winner}, where the following small-scale parameters are reciprocal:
\begin{itemize}
\item The angles $\theta_{m, \text{ZOA}}, \theta_{m, \text{ZOD}}, \varphi_{m, \text{AOA}}$, and $\varphi_{m, \text{AOD}}$.
\item The delay and the power of the path.
\item The cross polarization ratio $\kappa _{m}$.
\end{itemize}
The non-reciprocal channel parameters are listed below.
\begin{itemize}
\item The wavelength of the center frequency $\lambda_0$ and the related phase terms in Eq. (\ref{Eq:channelModel}).
\item Path loss factors for UL and DL.
\item The initial phases $\left\{ {\Phi _{m}^{\theta \theta },\Phi _{m}^{\theta \varphi },\Phi _{m}^{\varphi \theta },\Phi _{m}^{\varphi \varphi }} \right\}$, which are independent for UL and DL channels.
\end{itemize}

The topology of antennas at the base station is a UPA with $N_v$ rows and $N_h$ columns of antenna elements. The horizontal and vertical antenna spacings are denoted by ${D_h}$ and ${D_v}$ respectively. The total number of base station antennas is $N_t = N_v N_h N_p$, where $N_p = \{1, 2\}$ is the number of polarizations. For ease of exposition, we assume the antenna panel is in YZ plane of our coordinate system, with the first antenna element on the origin. The indices of the antennas are arranged the same way as in \cite{yin2020JSAC}, i.e., we order the antennas column by column. According to the channel model Eq. (\ref{Eq:channelModel}), we introduce a 3D steering vector, which is corresponding to a certain path with Azimuth departure angle $\varphi_\text{AOD}$ and Zenith departure angle $\theta_\text{ZOD}$.
\begin{equation}
{\bf{a}}({\theta}_\text{ZOD},{\varphi}_\text{AOD}) = {{\bf{a}}_h}({\theta}_\text{ZOD},{\varphi}_\text{AOD}) \otimes {{\bf{a}}_v}({\theta}_\text{ZOD}),
\end{equation}
where
\begin{align}
{{\bf{a}}_h}({\theta _{{\rm{ZOD}}}},{\varphi _{{\rm{AOD}}}}) = \left[ {\begin{array}{*{20}{c}}
1\\
{{e^{j2\pi {{D_h}\sin ({\theta _{{\rm{ZOD}}}})\sin ({\varphi _{{\rm{AOD}}}})}/\lambda _0}}}\\
 \vdots \\
{{e^{j2\pi {({N_h} - 1){D_h}\sin ({\theta _{{\rm{ZOD}}}})\sin ({\varphi _{{\rm{AOD}}}})}/{{\lambda _0}}}}}
\end{array}} \right],
\end{align}
and
\begin{equation}
{{\bf{a}}_v}({\theta _{{\rm{ZOD}}}}) = \left[ {\begin{array}{*{20}{c}}
1\\
{{e^{j2\pi \frac{{{D_v}\cos ({\theta _{{\rm{ZOD}}}})}}{{{\lambda _0}}}}}}\\
 \vdots \\
{{e^{j2\pi \frac{{({N_v} - 1){D_v}\cos ({\theta _{{\rm{ZOD}}}})}}{{{\lambda _0}}}}}}
\end{array}} \right].
\end{equation}
For notational simplicity, we temporarily assume the UE has one antenna. However the generalization from single-antenna UE to multiple-antenna UE is straightforward. In fact, the simulations in Sec. \ref{sec:numericalResult} are mostly carried out with multiple-antenna UEs.
Denote the number of subcarriers in UL and DL as $N_f^{\rm{(U)}}$ and $N_f^{\rm{(D)}}$ respectively.
Let $\vh^\text{(U)}(f, t) \in {\mathbb{C}^{N_t \times 1}}$ and $\vh^\text{(D)}(f, t) \in {\mathbb{C}^{N_t \times 1}}$ denote respectively the UL and DL channels between all base station antennas and the UE antenna at time $t$ and frequency $f$. We write the UL and DL channels at all subcarriers in a matrix form:
\begin{align}
{{\bf{H}}^{\text{(U)}}}(t) &\triangleq [\begin{array}{*{20}{c}}
{\vh^{\text{(U)}}(f_1, t)}&{\vh^{\text{(U)}}(f_2, t)}& \cdots &{{\vh^\text{(U)}(f_{N_f^{\rm{(U)}}}, t)}} \end{array}],\label{Eq:Hut} \\
{{\bf{H}}^{\text{(D)}}}(t) &\triangleq [\begin{array}{*{20}{c}}
{\vh^{\text{(D)}}(f_1, t)}&{\vh^{\text{(D)}}(f_2, t)}& \cdots &{{\vh^\text{(D)}(f_{N_f^{\rm{(D)}}}, t)}}
\end{array}], \label{Eq:Hdt}
\end{align}
where $f_i$ is the frequency of the $i$-th subcarrier of either UL or DL channel.

Depending on the context, a superscript $\text{X} = \{ \text{U}, \text{D}\}$ is introduced to simplify the presentation. According to the multipath model in Eq. (\ref{Eq:channelModel}), we may further write
\begin{equation}\label{Eq:HxtACB}
{{\bf{H}}^{(\text{X})}}(t) = {\bf{A}_\text{blk}^{(\text{X})}}{{\bf{C}}^{(\text{X})}}(t){{\bf{B}}^{(\text{X})}},
\end{equation}
where ${\bf{A}^{(\text{X})}_\text{blk}} = {\mathop{\rm diag}\nolimits} \{ {\bf{A}^{(\text{X})},...,\bf{A}^{(\text{X})}}\} \in {\mathbb{C}^{N_t \times N_p M}} $ is a block diagonal matrix with $\mA^{(\text{X})} \in {\mathbb{C}^{N_h N_v \times M}}$ composed of $M$ 3D steering vectors:
\begin{small}
\begin{equation}
{\bf{A}}^{\rm{(X)}} \buildrel \Delta \over = \left[ {\begin{array}{*{20}{c}}
{{\bf{a}}^{\rm{(X)}}({\theta _{1,{\rm{ZOD}}}},{\phi _{1,{\rm{AOD}}}})}& \cdots &{{\bf{a}}^{\rm{(X)}}({\theta _{M,{\rm{ZOD}}}},{\phi _{M,{\rm{AOD}}}})}
\end{array}} \right],
\end{equation}
\end{small}
and
\begin{equation}
{\bf{B}}^{\rm{(X)}} \buildrel \Delta \over = \left[ {\begin{array}{*{20}{c}}
{{\bf{b}}^{\rm{(X)}}({\tau _1})}&{{\bf{b}}^{\rm{(X)}}({\tau _2})}& \cdots &{{\bf{b}}^{\rm{(X)}}({\tau _M})}
\end{array}} \right]^T,
\end{equation}
with ${{\bf{b}}^{\rm{(X)}}({\tau _m})}$, $(m = 1, \cdots, M)$ being the delay response vector of the $m$-th path, which is defined as
\begin{equation}\label{Eq:btao}
\mbox{\small$\displaystyle
{\bf{b}}^{\rm{(X)}}({\tau _m}) = \left[ \begin{array}{*{20}{c}}
{{e^{ -j2\pi {f_1}{\tau _m}}}}&{{e^{ -j2\pi {f_2}{\tau _m}}}}& \cdots &{{e^{ -j2\pi {f{{N_f^{\text{(X)}}}}{\tau _m}}}}}
\end{array} \right]^T$}.
\end{equation}
The matrix ${\bf{C}}^{(\text{X})}(t) \in \mathbb{C}^{N_p M \times M}$ is defined as
\begin{align}\label{Eq:Cxt}
{{\bf{C}}^{({\rm{X}})}}(t) = {\left[ {\begin{array}{*{20}{c}}
{{{\bf{C}}_1}^{({\rm{X}})}(t)}& \cdots &{{{\bf{C}}_{{N_p}}}^{({\rm{X}})}(t)}
\end{array}} \right]^T},
\end{align}
where ${{\bf{C}}_i}^{({\rm{X}})}(t), i = 1, \cdots, N_p$ is a diagonal matrix
\begin{align}\label{Eq:Cixt}
{{{\bf{C}}_i}^{({\rm{X}})}(t)} = {\mathop{\rm diag}\nolimits} \{ {c^{(\text{X})}_{i, 1}(t),...,{c^{(\text{X})}_{i, M}}(t)}\},
\end{align}
with its $m$-th diagonal element defined as
\begin{align}\label{Eq:cimt}
&c^{(\text{X})}_{i, m}(t) = {\left[ {\begin{array}{*{20}{c}}
{{F_{{\rm{rx,}}\theta }}\left( {{\theta _{m{\rm{,ZOA}}}},{\varphi _{m{\rm{,AOA}}}}} \right)}\\
{{F_{{\rm{rx,}}\varphi }}\left( {{\theta _{m{\rm{,ZOA}}}},{\varphi _{m{\rm{,AOA}}}}} \right)}
\end{array}} \right]^T} \nonumber \\
&\left[ {\begin{array}{*{20}{c}}
{\exp \left( {j\Phi _m^{\theta \theta }} \right)}&{\sqrt {\kappa _m^{ - 1}} \exp \left( {j\Phi _m^{\theta \varphi }} \right)}\\
{\sqrt {{\kappa _m}^{ - 1}} \exp \left( {j\Phi _m^{\varphi \theta }} \right)}&{\exp \left( {j\Phi _m^{\varphi \varphi }} \right)}
\end{array}} \right] \nonumber \\
&\left[ {\begin{array}{*{20}{c}}
{{F_{{\rm{tx}},\theta }}\left( {{\theta _{m{\rm{,ZOD}}}},{\varphi _{m{\rm{,AOD}}}}} \right)}\\
{{F_{{\rm{tx}},\varphi }}\left( {{\theta _{m{\rm{,ZOD}}}},{\varphi _{m{\rm{,AOD}}}}} \right)}
\end{array}} \right]\exp \left( {j2\pi \frac{{\hat r_{{\rm{rx,}}m}^T\bar v}}{{{\lambda _0}}}t} \right).
\end{align}
The vectorized channel representation of Eq. (\ref{Eq:HxtACB}) is written as
\begin{align}\label{Eq:uhxt}
{\underline{\bf{h}}^{(\text{X})}}(t) \triangleq  \operatorname{vec} \left({{\bf{H}}^{(\text{X})}}(t)\right) = {{\bf{B}}^{(\text{X})T}} \otimes {\bf{A}_\text{blk}^{(\text{X})}}\operatorname{vec} \left( {{\bf{C}}^{(\text{X})}}(t)\right).
\end{align}

Denote the channel covariance matrix in spatial domain and frequency domain as
\begin{align}\label{Eq:Rxs}
{{\bf{R}}^{{\rm{(X,S)}}}} = \mathbb{E} \left\{ {{\bf{H}}^{\text{(X)}}}(t){{\left( {{\bf{H}}^{\text{(X)}}}(t) \right)}^H} \right\},
\end{align}
and
\begin{align}\label{Eq:Rxf}
{{\bf{R}}^{{\rm{(X,F)}}}} = \mathbb{E} \left\{ {{\left( {{\bf{H}}^{\text{(X)}}}(t) \right)}^T} \left({{\bf{H}}^{\text{(X)}}}(t)\right)^* \right\}.
\end{align}
respectively, where the expectation is taken over time.
We also denote the joint spatial-frequency channel covariance matrix as
\begin{align}\label{Eq:RxJ}
{{\bf{R}}^{{\rm{(X,J)}}}} = \mathbb{E} \left\{ {{\left( {\underline{\bf{h}}^{\text{(X)}}}(t) \right)}} \left({\underline{\bf{h}}^{\text{(X)}}}(t)\right)^H \right\}.
\end{align}
Note that the base station may obtain the DL channel covariance matrix by some existing methods, e.g., the projection method in a Hilbert space by exploiting the angular reciprocity of UL and DL channels \cite{miretti2018}.

\section{The Rank of the Covariance Matrix}\label{sec:covaRank}

In this section, we investigate the ranks of the channel covariance matrices of a UPA, for the cases of 1) the narrow-band spatial covariance matrix, 2) the frequency covariance matrix, and 3) the wideband joint spatial-frequency covariance matrix.
we demonstrate that the channel covariance matrix of a large-scale UPA with antenna spacing no larger than half wavelength always exhibits a low-rankness property no matter what the angular distribution is. This property shows that the spatial channel vector lives in a reduced space and thus can be compressed using the second-order channel statistics.
For ease of exposition, we derive the ranks for a single polarization. The results can be readily generalized to cross-polarization setting by simply multiplying the expressions of the ranks by 2.
We start with a simple setting where the multipath angles are randomly distributed within a closed region. For notational simplicity, we drop the subscripts ``ZOD" and ``AOD"  and do not specify UL or DL in this section. The Zenith angle of departure $\theta$ and the Azimuth angle of departure $\varphi$ are distributed with an arbitrary probability density function that is  non-zero inside the bounded and convex support $\zeta(\theta, \varphi)$. The boundary is defined as follows: $\theta \in [\theta^{\text{min}}, \theta^{\text{max}}]$, and for a given Zenith angle $\theta$, the bound of the Azimuth angle is $[\varphi^{\text{min}}(\theta), \varphi^{\text{max}}(\theta)]$. Consider a certain carrier frequency with wavelength $\lambda$. The rank of the spatial covariance matrix ${{\bf{R}}^{{\rm{(S)}}}}$, is characterized in Theorem \ref{theoSpatialRank} for massive MIMO regime.

\begin{theorem}\label{theoSpatialRank}
When $N_h$ and $N_v$ are large, the rank of the spatial channel covariance matrix ${{\bf{R}}^{{\rm{(S)}}}}$ is given by
\begin{equation}\label{Eq:theoRank}
\frac{\text{rank}\{{{\bf{R}}^{{\rm{(S)}}}} \}}{N_h N_v} = \rho^{{\rm{(S)}}},
\end{equation}
where the ratio $\rho^{{\rm{(S)}}}$ is defined as
\begin{align}\label{Eq:rhoS}
\rho^{{\rm{(S)}}} \defi \frac{{{D_h}{D_v}}}{{{\lambda ^2}}}\int\limits_{{\theta ^{\min }}}^{{\theta ^{\max }}} {{{\sin }^2}(\theta )\left( {\sin \left( {{\varphi ^{\max }}(\theta )} \right) - \sin \left( {{\varphi ^{\min }}(\theta )} \right)} \right)d\theta }.
\end{align}
\end{theorem}
\begin{proof}
\quad \emph{Proof:} The proof can be found in Appendix \ref{proof:theoSpatialRank}.
\end{proof}

For the special case where the ZOD and AOD are mutually independent and distributed within their own intervals, i.e., $\theta \in [\theta^{\text{min}}, \theta^{\text{max}}]$, $\varphi \in [\varphi^{\text{min}}, \varphi^{\text{max}}]$, as shown in Fig. \ref{fig:AngleRange}, we may readily obtain the rank ratio of the spatial covariance matrix:
\begin{align}\label{Eq:rhoSInd}
\rho^{{\rm{(S)}}} &= \frac{{{D_h}{D_v}}}{{{\lambda ^2}}}\left( {\sin ({\varphi ^{\max }}) - \sin ({\varphi ^{\min }})} \right) \nonumber \\
& \left( {\frac{1}{2}({\theta ^{\max }} - {\theta ^{\min }}) - \frac{1}{4}\left( {\sin (2{\theta ^{\max }}) - \sin (2{\theta ^{\min }})} \right)} \right).
\end{align}

\begin{figure}[h]
  \centering
  \includegraphics[width=3.0in]{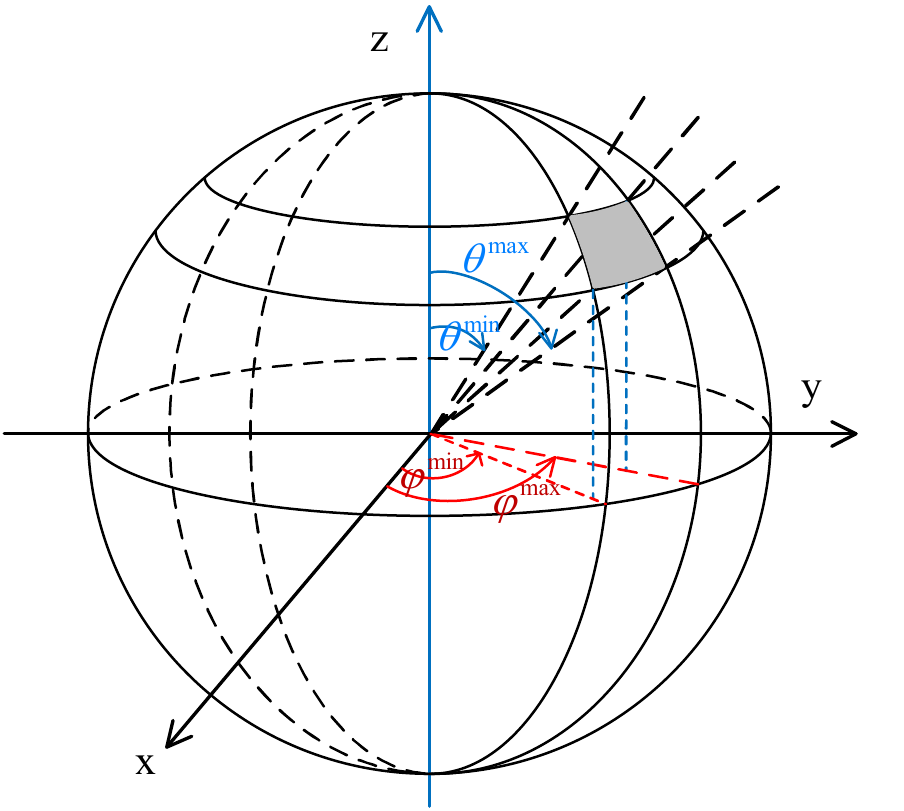}\\
  \caption{An example of the bounded angular support.} \label{fig:AngleRange}
\end{figure}

Furthermore, if the ZOD and AOD spread the full angular range, e.g., $\theta$ and $\varphi$ are distributed with non-zero probability in all directions, then the rank ratio of the spatial channel covariance matrix is upper-bounded by
\begin{align}\label{Eq:rhoSAll}
\rho^{{\rm{(S)}}} &= \frac{{{D_h}{D_v}}}{{{\lambda ^2}}} \pi,
\end{align}
or equivalently
\begin{align}\label{Eq:rankRS}
{\text{rank}\{{{\bf{R}}^{{\rm{(S)}}}} \}} = L_h L_v \pi,
\end{align}
where $L_h$ and $L_v$ denote the normalized apertures of the UPA in horizontal and vertical directions respectively:
\begin{align}
L_h \defi D_h N_h/\lambda , \quad L_v \defi D_v N_v/\lambda.
\end{align}
Eq. (\ref{Eq:rankRS}) indicates that the rank of the spatial covariance matrix is determined by the area of the ellipse circumscribing the rectangle of size $L_v \times L_h$. In the practical cases that the antenna spacing is no larger than half-wavelength, such an observation means ${{\bf{R}}^{{\rm{(S)}}}}$ always have low-rankness property, and the rank ratio will never exceed $\pi/4$ under arbitrary angular distributions.

Interestingly, Eq. (\ref{Eq:rhoSAll}) coincides with the Degrees of Freedom (DoF) in \cite{Pizzo2020} for a UPA in an isotropic scattering environment. In fact, Theorem \ref{theoSpatialRank} is a generalized result of the DoF for the case that the angles have a certain bounded support instead of unbounded.

Note that Theorem \ref{theoSpatialRank} describes the dimensionality of the signal subspace for a simple setting that the multipath angles exists inside one bounded support. Corollary \ref{coro:RankDisjoint} is an extension to the case that the angles are distributed within several disjoint angular sub-supports, which is more realistic. Denote the number of disjoint sub-supports as $Q$, and the boundary of the $q$-th sub-support as
\begin{align}
\zeta_q \defi \left\{ (\theta, \varphi) | \theta \in [\theta_q^\text{min}, \theta_q^\text{max}], \varphi \in [\varphi_q^\text{min}(\theta), \varphi_q^\text{max}(\theta)]\right\}.
\end{align}
The sub-supports are mutually disjoint, i.e., $\zeta_i \cap \zeta_j = \emptyset, \forall i \neq j$. Consider a certain UE with the angular support $\zeta_u$ being the union of all $Q$ sub-supports, i.e., $\zeta_u = \zeta_1 \cup \zeta_2 \cup, \cdots, \cup \zeta_Q$, and we assume the angular power spectrum (APS) is non-zero over the support $\zeta_u$.
\begin{corollary}\label{coro:RankDisjoint}
When $N_h$ and $N_v$ are large, the rank of the spatial channel covariance matrix ${{\bf{R}}^{{\rm{(S)}}}}$ is given by
\begin{equation}
\frac{\text{rank}\{{{\bf{R}}^{{\rm{(S)}}}} \}}{N_h N_v} = \rho^{{\rm{(S)}}},
\end{equation}
where the ratio $\rho^{{\rm{(S)}}}$ is defined as
\begin{small}
\begin{align}
\rho^{{\rm{(S)}}} \defi \frac{{{D_h}{D_v}}}{{{\lambda ^2}}}\sum\limits_{q = 1}^Q {\int\limits_{\theta _q^{\min }}^{\theta _q^{\max }} {{{\sin }^2}(\theta )\left( {\sin \left( {\varphi _q^{\max }(\theta )} \right) - \sin \left( {\varphi _q^{\min }(\theta )} \right)} \right)d\theta } }.
\end{align}
\end{small}
\end{corollary}
\begin{proof}
\quad \emph{Proof:} The proof is a simple extension of Appendix \ref{proof:theoSpatialRank} and thus omitted.
\end{proof}

Theorem \ref{theoSpatialRank} and Corollary \ref{coro:RankDisjoint} show that the spatial covariance matrix has a reduced subspace and only a limited number of eigenvectors have corresponding non-negligible eigenvalues. As a result, the channel matrix is compressible using the dominant eigenvectors of ${\bf{R}}^{{\rm{(S)}}}$.

We then show the frequency covariance matrix has a similar low-rankness property, which can be exploited to reduce the dimension of the channel in frequency domain. Suppose the power delay profile (PDP) is strictly non-zero inside $Q$ delay intervals $\xi_1, \xi_2, \cdots, \xi_Q$. The range of the $q$-th interval is $[\tau_q^\text{min}, \tau_q^\text{max}]$, $q = 1, \cdots, Q$. Denote the subcarrier spacing as $\Delta f$ and the number of subcarriers as $N_f$. Proposition \ref{prop:RankFreq} describes the rank of the corresponding spatial covariance matrix $\mR^\text{(F)}$.
\begin{proposition}\label{prop:RankFreq}
When $N_f$ is large, the rank of the frequency channel covariance matrix $\mR^\text{(F)}$ is given by
\begin{equation}
\frac{\text{rank}\{\mR^\text{(F)}\}}{N_f}= \rho^{{\rm{(F)}}},
\end{equation}
where the ratio $\rho^{{\rm{(F)}}}$ is defined as
\begin{align}
\rho^{{\rm{(F)}}} \defi \min \left\{ \sum\limits_{q = 1}^Q \Delta f {\left( {\tau _q^{\max } - \tau _q^{\min }} \right)}, 1 \right\}.
\end{align}
\end{proposition}
\begin{proof}
\quad \emph{Proof:} The proof is readily obtained by applying some modifications of Lemma 2 in \cite{yin:13}, which was originally dedicated to the rank of the spatial covariance matrix for large-scale ULA. Thus we skip the detailed proof here.
\end{proof}

Proposition \ref{prop:RankFreq} indicates that as the bandwidth increases, the frequency covariance matrix is also rank-deficient due to the limited delay distributions. As a result, the channel matrix can be compressed in frequency domain using the dominant eigenvectors of the covariance matrix $\mR^\text{(F)}$.

We future study the rank of the joint spatial-frequency covariance matrix $\mR^\text{(J)}$. Denoting now the number of angle-delay sub-supports as $Q$, the boundary of the $q$-th angle-delay sub-support is
\begin{align}\label{Eq:etaq}
\eta_q &\defi \left\{ (\theta, \varphi, \tau) | \theta \in [\theta_q^\text{min}, \theta_q^\text{max}], \varphi \in [\varphi_q^\text{min}(\theta), \varphi_q^\text{max}(\theta)], \right. \nonumber \\
& \quad \left.\tau \in [\tau_q^\text{min}(\theta, \varphi), \tau_q^\text{max}(\theta, \varphi)]\right\}.
\end{align}
The union of all the sub-supports is defined as
\begin{align}\label{Eq:etau}
\eta_u = \eta_1 \cup \eta_2 \cup, \cdots, \cup \eta_Q.
\end{align}
The probability density function of the multipath is uniformly non-zero within each sub-support, which is assumed convex for ease of exposition. We also make an assumption that the longest delay $\tau^\text{max}$ satisfies $\tau^\text{max} \leq 1/{\Delta f}$. Taking a 15 kHz subcarrier spacing for example, it means the paths are not longer than 20 km, which is reasonable. The rank of $\mR^\text{(J)}$ is shown in Theorem \ref{theoJointRank}.
\begin{theorem}\label{theoJointRank}
When $N_h$, $N_v$, and $N_f$ are large, the rank of the joint spatial-frequency channel covariance matrix ${{\bf{R}}^{{\rm{(J)}}}}$ is given by
\begin{equation}\label{Eq:theoJointRank}
\frac{\text{rank}\{{{\bf{R}}^{{\rm{(J)}}}} \}}{N_h N_v N_f} = \rho^{{\rm{(J)}}},
\end{equation}
where the ratio $\rho^{{\rm{(J)}}}$ is defined as
\begin{align}\label{Eq:rhoJ}
{\rho ^{({\rm{J}})}}&= \frac{{{D_h}{D_v}\Delta f}}{{{\lambda ^2}}} \sum\limits_{q = 1}^Q \int\limits_{\theta _q^{\min }}^{\theta _q^{\max }} \int\limits_{\varphi _q^{\min }(\theta )}^{\varphi _q^{\max }(\theta )} \left\{ {\sin }^2(\theta )\cos (\varphi ) \right. \nonumber \\
& \quad \left. \left( {{\tau_q ^{\max }}(\theta ,\varphi ) - {\tau_q ^{\min }}(\theta ,\varphi )} \right) \right\} d\varphi d\theta .
\end{align}
\end{theorem}
\begin{proof}
\quad \emph{Proof:} The proof can be found in Appendix \ref{proof:theoJointRank}.
\end{proof}

Note that the rank expressions for the spatial and frequency covariance matrices may serve as an upper-bound. In practice, the multipath do not have strictly non-zero APS or PDP over one or several intervals. In some settings, the angles and delays are modeled as discrete values that are constant over a relatively long time period, as in the clustered delay line (CDL) model of 3GPP \cite{3gpp:38.901} for example. In such cases, the rank ratio ${\rho ^{({\rm{S}})}}$, ${\rho ^{({\rm{F}})}}$, and ${\rho ^{({\rm{J}})}}$ converge to zero for any finite number of multipaths.

\section{A Partial Channel Reciprocity-based Codebook}\label{sec:PCR}
In this section, we show the details of our proposed partial channel reciprocity-based codebook.
Since the multipath angle-delay distribution of the UL and DL are reciprocal \cite{3gpp:36.897} \cite{Zhong2020}, there exists a partial reciprocity in terms of the signal subspace for UL and DL wideband channel covariance matrices.
Our method makes use of such a partial reciprocity and the low-dimensional structure of the statistical channel information in both spatial and frequency domains. The DL channel sounding is designed based on the prior knowledge of the channel covariance matrices in such a way that only the signal space corresponding to the non-negligible eigenvectors of the joint spatial-frequency covariance matrices is captured and measured by the UE, while the null space of the covariance matrices is automatically ignored. This will help reduce the feedback overhead and the complexity at the UE side.
More specifically, the training signal is sent according to the dominant eigenvectors of the covariance matrices.

For ease of exposition, we drop the superscript $\text{(U)}$ or $\text{(D)}$ and consider the downlink by default, unless otherwise notified.
Denote the eigen-value decomposition of the joint spatial-frequency channel covariance matrices as
\begin{align}
{{\bf{R}}^{{\rm{(J)}}}} &= {\mU}^{{\rm{(J)}}} {\mathbf{\Sigma}^{{\rm{(J)}}}} {\mU}^{{{\rm{(J)}}}H}, \label{Eq:Rj_EVD}
\end{align}
where ${\mU}^{{\rm{(J)}}}$  contains the eigen-vectors of ${{\bf{R}}^{{\rm{(J)}}}}$:
\begin{align}\label{Eq:EigsJ}
{\mU}^{{\rm{(J)}}} &= \left[{\vu}_1^{{\rm{(J)}}}, {\vu}_2^{{\rm{(J)}}}, \cdots, {\vu}_{N_t N_f}^{{\rm{(J)}}}\right].
\end{align}
${\mathbf{\Sigma}^{{\rm{(J)}}}}$ is a diagonal matrix with its diagonal entries (arranged in non-increasing order) being the eigenvalues of ${{\bf{R}}^{{\rm{(J)}}}}$. As shown in Sec. \ref{sec:covaRank}, ${{\bf{R}}^{{\rm{(J)}}}}$ has a low-rankness property and only a small number of its eigenvalues are significant, while the rest are negligible. We exploit this property and propose the joint spatial-frequency precoder for CSI-RS, which is aligned with the dominant eigenvectors of the covariance matrix ${{\bf{R}}^{{\rm{(J)}}}}$. Denote the total number of antenna ports for CSI-RS by $N_a$. The joint spatial-frequency precoder $\vw_n \in {\mathbb{C}^{N_t N_f \times 1}}$ for port $n$ of CSI-RS is
\begin{align}\label{Eq:wn}
    \vw_n \defi \left({\vu}_n^{\rm{(J)}}\right)^*, n = 1, 2, \cdots, N_a.
\end{align}
Note that the antenna port here is a generalized concept of the 5G \cite{dahlman2020}. The  antenna port $n$ ($n = 1, \cdots, N_a$) for CSI-RS corresponds to the specific reference signals with the $n$-th spatial-frequency precoder.
In order to better illustrate the idea of our method, we show the joint operation of the base station and the UE step by step. We introduce here the subband concept, which means a group of consecutive subcarriers in frequency domain. Depending on implementation, the bandwidth of a subband is equal to the width of one or several Resource Blocks (RBs), and it is smaller than the coherence bandwidth. In other words, the frequency response of the channel is flat within this interval and only one channel estimate is needed for each subband. With some misuse of notation, $N_f$ denotes either the number of subcarriers or the number of subbands.

For a certain port $n$, the precoder $\vw_n$ can be unvectorized into $N_f$ vectors as below:
\begin{align}
\mW_n \defi [\vw_{n,1}, \vw_{n,2}, \cdots, \vw_{n,N_f}] = \text{unvec}\{\vw_n\},
\end{align}
where $\vw_{n, k} \in \mathbb{C}^{N_t \times 1}$, $k = 1, \cdots, N_f$, is the precoder of CSI-RS for the $k$-th subband. For notational simplicity, we assume the training sequences for each CSI-RS port is identical in different subbands. Denote the reference signal for the $n$-th CSI-RS port as ${{\bf{x}}_n} \in \mathcal{C}^{N_x \times 1}$, which is distributed in a set of $N_x$ neighboring Resource Elements (REs). The reference signals of different antenna ports may share the same set of REs by code-domain sharing (CDM) with orthogonal patterns, or Frequency-domain sharing (FDM) with different subcarriers, or Time-domain sharing (TDM) with different OFDM symbols. For antenna port $n$, $1 \leq n \leq N_a$, the transmitted jointly precoded reference signal is ${{\bf{x}}_n}\vw^T_{n, k}$ at the $k$-th subband. The corresponding received signal of this port at the $k$-th subband is
\begin{align}\label{Eq:ynkt}
{{\bf{y}}_{n,k}}(t) = {{\bf{x}}_n}{{{\vw}}^T_{n,k}}{{\bf{h}}}({f_k},t) + {{\bf{n}}_{n,k}}(t),
\end{align}
where ${{\bf{h}}}({f_k},t) \in \mathbb{C}^{N_t \times 1}$ denotes the DL channel at subband $k$, and ${{\bf{n}}_{n,k}}(t) \in \mathbb{C}^{N_x \times 1}$ is the noise. Note that till now we have only completed part of the precoding, which is the operation ${{{\vw}}^T_{n,k}}{{\bf{h}}}({f_k},t)$. The joint spatial-frequency precoding needs a further action from UE side, which is a summation in frequency domain shown below:
\begin{align}
{{\bf{y}}_n}(t) &= \sum\limits_{k = 1}^{N_f} {{\bf{y}}_{n,k}}(t) \label{Eq:ynt_sigma}\\
& = {{\bf{x}}_n}\vw_n^T {\underline{\bf{h}}}(t) + {{\bf{n}}_n}(t),
\end{align}
where
\begin{align}
{{\bf{n}}_n}(t) = \sum\limits_{k = 1}^{N_f} {{{\bf{n}}_{n,k}}(t)} \in \mathcal{C}^{N_x \times 1}.
\end{align}
${\underline{\bf{h}}}(t)$ is the vectorized wideband channel as defined in Eq. (\ref{Eq:uhxt}). The $\sum$ operation in Eq. (\ref{Eq:ynt_sigma}) is done by the UE. Rewriting the received signal for all $N_a$ ports in matrix form, we have
\begin{align}\label{Eq:Yt}
{\bf{Y}}(t) = {\bf{Xg}}(t) + {\bf{n}}(t),
\end{align}
where ${\bf{X}} \in \mathcal{C}^{N_a N_x \times N_a}$ is a block matrix containing all the training sequences
\begin{align}\label{Eq:X}
{\bf{X}} \defi \left[ {\begin{array}{*{20}{c}}
{{{\bf{x}}_1}}&{}&{}\\
{}& \ddots &{}\\
{}&{}&{{{\bf{x}}_{{N_a}}}}
\end{array}} \right].
\end{align}
${\bf{g}}(t) \in \mathcal{C}^{N_a \times 1}$ is the spatial-frequency jointly precoded effective channel:
\begin{align}\label{Eq:gt}
{\bf{g}}(t) \defi {\left[ {\begin{array}{*{20}{c}}
{{{{g}}_1}(t)}&{{{{g}}_2}(t)}& \cdots &{{{{g}}_{{N_a}}}(t)}
\end{array}} \right]^T},
\end{align}
with the effective scalar channel for port $n$ defined as
\begin{align}\label{Eq:gnt}
g_n(t) \defi \vw_n^T {\underline{\bf{h}}}(t).
\end{align}
Note that $g_n(t)$ can be regarded as a projection of the wideband channel onto the linear space determined by the eigen-mode ${\vu}_n^{{\rm{(J)}}}$ of the joint spatial-frequency covariance matrix ${{\bf{R}}^{{\rm{(J)}}}}$. The UE then obtains an estimate of the projection $\hat{g}_n(t)$ based on the known pilot sequence $\vx_n$ by simply performing an estimation of the effective channel ${{{\vw}}^T_{n,k}}{{\bf{h}}}({f_k},t)$ of each subband $k$ and summing them up over all $N_f$ subbands. Alternatively, the UE may first perform a summation of the received reference signal as in Eq. (\ref{Eq:ynt_sigma}), and then estimate the effective scalar channel $g_n(t)$, provided that the training sequence for a certain port is identical for all subbands.

Finally, the UE feeds back the projections $\hat{g}_n(t), n = 1, \cdots, N_a$ after a proper quantization, e.g., elementwise quantization \cite{3gpp:38.214} or more sophisticated schemes like random vector quantization (with higher complexity), to the base station. The DL channel is reconstructed as
\begin{align}\label{Eq:huhatDt_PCR}
{\underline{\bf{\hat h}}}(t) = \sum\limits_{n = 1}^{{N_a}} {{{\hat g}_n}(t)} {\bf{w}}_n^*.
\end{align}

We point out the fundamental difference between the classical precoding and the proposed spatial-frequency joint precoding here.
The classical precoding is to apply a beamforming weight vector (i.e., the precoder) on the base station antennas while they are transmitting signals. The UE receives the superimposed signals of all transmit antennas. Such a process is done by the base station alone.
On the contrary, the joint spatial-frequency precoding functions in a different manner, as it entails a joint processing of the base station and the UE. More precisely, the base station performs an elementwise multiplication, while the UE makes a summation. This is because in an OFDM system, the signal transmitted in all subcarriers are orthogonal to one another, which prohibits the summation of signals with orthogonal frequencies over the air as in the spatial-domain precoding case. As a result, the process consists of three steps, 1) Each base station antenna multiplies the CSI-RS at the pilot-carrying resource elements by the corresponding entry of $\vw_n$; 2) Upon receiving the precoded CSI-RS, the UE estimates the effective channels for all pilot-carrying subcarriers; 3) the UE sums up these effective channels over all pilot-carrying subcarriers.

The proposed PCR codebook scheme is summarized in Algorithm \ref{Alg:PartialReciprocityCodebook}.
\begin{algorithm}
\caption{Partial Channel Reciprocity based Codebook Scheme}
\begin{algorithmic}[1]\label{Alg:PartialReciprocityCodebook}

\STATE{The base station obtains the joint spatial-frequency precoders $\vw_n$ $n = 1, \cdots, N_a$ by eigenvalue decomposition of the covariance matrix $\mR^\text{(J)}$; }

\STATE{The base station applies the precoders when transmitting CSI-RS;}

\STATE{For each antenna port, the UE estimates the effective channels in all $N_f$ subbands and sum them up to obtain an estimate of $g_n(t)$;}

\STATE{The UE feeds back the scalar values $g_n(t), n = 1, \cdots, N_a$;}

\STATE{The base station reconstructs the DL CSI with Eq. (\ref{Eq:huhatDt_PCR});}

\end{algorithmic}
\end{algorithm}

Note that a prerequisite of this algorithm is that the DL channel covariance matrix $\mR^\text{(J)}$ is known by the base station. In practice, the UL channel covariance matrices are easily obtained by channel estimation using UL pilot (SRS, Sounding Reference Signal). The DL channel covariance matrix may be obtained by transforming the UL covariance matrix to DL using the Hilbert space projection method in \cite{miretti2018}.
Since the second-order statistics are slow-varying, their eigenvectors may not have to be updated very frequently. What is more, it is also possible to replace the DL channel covariance matrices with the UL ones without any transformation. This only leads to minor performance drops, which will be shown in Sec. \ref{sec:numericalResult}.

Overall, this proposed training and feedback scheme greatly reduces the computation burden at the UE side, as the spatial and frequency-domain precoders are computed and applied by the base station alone. Such an operation is transparent to the UE, which does not have to know the precoders themselves. In order to obtain coefficients $\hat{\bf{g}}(t)$ for feedback, the UE only has to perform an estimation of the precoded effective channels and an addition in frequency domain.

The complexity of this scheme is now analyzed. The main operation for the UE is the CSI-RS based channel estimation and a summation. Since the channel estimation is done by correlating the received signal with the pilot sequence of length $N_x$ for all $N_f$ subbands and $N_a$ ports, it requires $N_x N_f N_a$ floating point operations (FLOPs). The summation in frequency domain takes $N_f N_a$ FLOPs. Therefore, the total complexity order for the UE is $\mathcal{O}(N_x N_f N_a)$. Given that the orders of magnitude for $N_x$, $N_f$, and $N_a$ are all ten for practical 5G systems, the complexity of our proposed method is quite low for UEs.

We now analyze the asymptotic performance of our proposed PCR method. We derive the amount of feedback required by PCR method to achieve asymptotically error-free CSI feedback. Note that the number of antenna ports for CSI-RS, $N_a$, is also the number of scalar coefficients needed for feedback in our scheme for the case of single-antenna UE. The boundaries of the angle and delay distribution is the same as in Eq. (\ref{Eq:etaq}) and Eq. (\ref{Eq:etau}). The result is shown in Theorem \ref{theoFeedbackAmount}.
\begin{theorem}\label{theoFeedbackAmount}
The quantization error of the proposed PCR scheme yields
\begin{equation}\label{Eq:theoFeedbackAmount}
\mathop {\lim }\limits_{{N_h},{N_v},{N_f} \to \infty } \frac{{\left\| {{\bf{\hat H}} - {\bf{H}}} \right\|_F^2}}{{\left\| {\bf{H}} \right\|_F^2}} = 0,
\end{equation}
under the condition that the number of scalar coefficients to feedback satisfies
\begin{align}
N_a \geq \rho^{{\rm{(J)}}} N_h N_v N_f,
\end{align}
where ${\bf{H}}$ is the wideband channel matrix and ${\bf{\hat H}}$ is the reconstructed channel matrix based on the feedback from the UE. $\rho^{{\rm{(J)}}}$ is the rank ratio defined in Eq. (\ref{Eq:rhoJ}).
\end{theorem}
\begin{proof}
\quad \emph{Proof:} The proof can be found in Appendix \ref{proof:theoFeedbackAmount}.
\end{proof}

Theorem \ref{theoFeedbackAmount} demonstrates that our PCR scheme only requires a small amount of feedback, since the rank ratio $\rho^{\rm{(J)}} \ll 1$. In fact, $\rho^{\rm{(J)}}$ is also the CSI compression ratio of our scheme compared to the method of feeding back the full channel matrix ${\bf{H}}$. For the case that angles and delays are modeled as constant over a long period as in the CDL model, the number of feedback coefficients $N_a$ does not scale with the number of antennas or the bandwidth. It is in fact a finite value that is always smaller than the number of paths. This is because in a practical model such as CDL, the multipaths exist in the form of clusters. And the path angles inside a certain cluster are close to each other, since they are reflected by one or several neighboring scatterers. Some examples in Sec. \ref{sec:numericalResult} will show that although the total number of paths is up to several hundreds, the number of feedback coefficients can be as small as 64 for the whole wideband massive MIMO channel.

\section{Low Complexity Alternatives}\label{sec:lowComplex}
The PCR codebook scheme in Algorithm \ref{Alg:PartialReciprocityCodebook} of Sec. \ref{sec:PCR} requires the eigen-decomposition of a large matrix $\mR^\text{(J)}$. Although this operation is not done in a real-time manner, the complexity is still high for a base station. In this section, we propose two low-complexity alternatives which have mild performance losses compared with the original PCR codebook scheme. Note that all the proposed schemes have the same complexity at the UE side.

\subsection{PCR codebook with spatial/frequency eigen basis}\label{sec:PCR-E}
We now propose an alternative - the PCR codebook with spatial/frequency eigen basis (PCR-E) to mitigate this impediment of high complexity at the base station. It is done by selecting the eigenvectors of the covariance matrices $\mR^\text{(S)} \in \mathbb{C}^{N_t \times N_t}$ and $\mR^\text{(F)} \in \mathbb{C}^{N_f \times N_f}$ separately and then form a joint spatial-frequency precoder. The design of a joint spatial-frequency precoder breaks down to a spatial-domain precoder and a frequency-domain precoder.

The eigen-value decomposition of the spatial and frequency channel covariance matrices are written as
\begin{align}
{{\bf{R}}^{{\rm{(S)}}}} &= {\mU}^{{\rm{(S)}}} {\mathbf{\Sigma}^{{\rm{(S)}}}} {\mU}^{{{\rm{(S)}}}H}, \label{Eq:Rs_EVD}\\
{{\bf{R}}^{{\rm{(F)}}}} &= {\mU}^{{\rm{(F)}}} {\mathbf{\Sigma}^{{\rm{(F)}}}} {\mU}^{{{\rm{(F)}}}H}, \label{Eq:Rf_EVD}
\end{align}
where ${\mU}^{{\rm{(S)}}}$ and ${\mU}^{{\rm{(F)}}}$ contain the eigen-vectors of ${{\bf{R}}^{{\rm{(S)}}}}$  and ${{\bf{R}}^{{\rm{(F)}}}}$ respectively:
\begin{align}\label{Eq:EigsS}
{\mU}^{{\rm{(S)}}} &= \left[{\vu}_1^{{\rm{(S)}}}, {\vu}_2^{{\rm{(S)}}}, \cdots, {\vu}_{N_t}^{{\rm{(S)}}}\right], \\
{\mU}^{{\rm{(F)}}} &= \left[{\vu}_1^{{\rm{(F)}}}, {\vu}_2^{{\rm{(F)}}}, \cdots, {\vu}_{N_f}^{{\rm{(F)}}}\right]. \label{Eq:EigsF}
\end{align}
${\mathbf{\Sigma}^{{\rm{(S)}}}}$ and ${\mathbf{\Sigma}^{{\rm{(F)}}}}$ are diagonal matrices with their diagonal entries (arranged in non-increasing order) being the eigenvalues of ${{\bf{R}}^{{\rm{(S)}}}}$ and ${{\bf{R}}^{{\rm{(F)}}}}$ respectively.
We once again exploit the low-rankness properties of ${{\bf{R}}^{{\rm{(S)}}}}$ and ${{\bf{R}}^{{\rm{(F)}}}}$ in this PCR-E codebook scheme. We build the spatial-domain precoder of the CSI-RS such that it is aligned with the dominant eigenvectors of the spatial-domain covariance matrix ${{\bf{R}}^{{\rm{(S)}}}}$. Meanwhile, the frequency-domain precoder is aligned with the dominant eigenvectors of ${{\bf{R}}^{{\rm{(F)}}}}$.
The joint spatial-frequency precoder $\vw_n \in {\mathbb{C}^{1 \times N_t N_f}}$ for port $n$ of CSI-RS is given as
\begin{align}
    \vw_n \defi \vf_n \otimes \vs_n,
\end{align}
where $\vf_n \in {\mathbb{C}^{1 \times N_f}}$ is the conjugate transpose of a certain eigenvector selected from ${\mU}^{{\rm{(F)}}}$ and $\vs_n \in {\mathbb{C}^{1 \times N_t}}$ is the conjugate transpose of a certain eigenvector from  ${\mU}^{{\rm{(S)}}}$.
One simple eigenvector selection method is to choose the eigenvectors corresponding to the greatest eigenvalues of ${{\bf{R}}^{{\rm{(S)}}}}$ and ${{\bf{R}}^{{\rm{(F)}}}}$. However, this method may not lead to the best performance.
The reason is that the multipath distribution of angles is not independent to the distribution of delays in most cases, e.g., when the paths are distributed in the form of clusters. ${{\bf{R}}^{{\rm{(S)}}}}$ only contains the multipath distribution information in angular domain, while ${{\bf{R}}^{{\rm{(F)}}}}$ only reflects the multipath delay distribution. If we choose a set of spatial precoders independently with the set of frequency precoders  and construct the joint precoders by all the combinations of the two sets, it will result in excessive consumption of training and feedback overhead. This is due to the fact that some constructed joint precoders may lead to very small coefficients, which are not necessary for feedback. For the sake of overhead reduction, we propose a selection criterion which selects the spatial precoders and frequency precoders jointly based on recent UL channel samples
${{\bf{H}}^{{\rm{(U)}}}}({t_1}), {{\bf{H}}^{{\rm{(U)}}}}({t_2}), \cdots, {{\bf{H}}^{{\rm{(U)}}}}({t_{N_c}})$.

Define the joint spatial-frequency eigen projection as
\begin{align}\label{Eq:Gu}
\mG(t) \defi {{\left({\bf{U}}^{({\rm{S}})}\right)^H}{{\bf{H}}^{{\rm{(U)}}}}({t})\left({{\bf{U}}^{(F)}}\right)^*}.
\end{align}
Then, the accumulated element-wise power matrix of the projection above is
\begin{align}\label{Eq:calG}
\bm{\mathcal{G}} \defi \sum\limits_{i = 1}^{{N_c}} \mG(t_i) \odot \left(\mG(t_i) \right)^*.
\end{align}
We select the positions of the greatest $N_a$ values in $\bm{\mathcal{G}}$ and denote the row index and column index of the $n$-th position as $r_n$ and $c_n$ respectively. Then, the spatial-domain precoder for the $n$-th CSI-RS port is the conjugate of the $r_n$-th eigenvector of ${{\bf{R}}^{{\rm{(S)}}}}$:
\begin{align}\label{Eq:sn}
\vs_n = \left({\vu}_{r_n}^{{\rm{(S)}}}\right)^*.
\end{align}
Similarly, the frequency-domain precoder for this port is
\begin{align}\label{Eq:fn}
\vf_n = \left({\vu}_{c_n}^{{\rm{(F)}}}\right)^*.
\end{align}
Note that among all $N_a$ ports, there are always some duplicates of the spatial-domain precoder or frequency-domain precoder for different ports. In other words, the following circumstance may take place:
\begin{align}
\exists i \neq j, \text{ s.t. } r_i = r_j \text{ or } c_i = c_j, \text{ for } i,j = 1, \cdots, N_a.
\end{align}
This is due to the fact that the accumulated power matrix $\bm{\mathcal{G}}$ may have some significant entries distributed in the same row or column. However, two different CSI-RS ports never have exactly the same spatial precoder and frequency precoder at the same time.
Regarding the antenna port $n$, $1 \leq n \leq N_a$, the transmitted jointly precoded reference signal is ${{\bf{x}}_n}\left( {{{\bf{f}}_n} \otimes {{\bf{s}}_n}} \right)$. The corresponding received signal at frequency $f_k$, a.k.a., the $k$-th subband, is
\begin{align}
{{\bf{y}}_{n,k}}(t) = {{\bf{x}}_n}{{{f}}_{n,k}}{{\bf{s}}_n}{\bf{h}}({f_k},t) + {{\bf{n}}_{n,k}}(t),
\end{align}
where ${{{f}}_{n,k}}$ is the $k$-th row of the frequency-domain precoder ${{\bf{f}}_{n}}$. After the frequency-domain summation at the UE side, the combined received signal is
\begin{align}
{{\bf{y}}_n}(t) &= \sum\limits_{k = 1}^{N_f} {\left( {{{\bf{x}}_n}{{{f}}_{n,k}}{{\bf{s}}_n}{{\bf{h}}}({f_k},t) + {{\bf{n}}_{n,k}}(t)} \right)} \label{Eq:ynt_sigma_PCRE}\\
&= {{\bf{x}}_n}\left( {{{\bf{f}}_n} \otimes {{\bf{s}}_n}} \right){\underline{\bf{h}}}(t) + {{\bf{n}}_n}(t) ,
\end{align}
where
\begin{align}
{{\bf{n}}_n}(t) = \sum\limits_{k = 1}^{N_f} {{{\bf{n}}_{n,k}}(t)} \in \mathcal{C}^{N_x \times 1}.
\end{align}
The $\sum$ operation in Eq. (\ref{Eq:ynt_sigma_PCRE}) is done by the UE.
Based on the training sequence $\vx_n$, we may obtain the spatial-frequency jointly precoded effective channel:
\begin{align}
{\bf{g}}(t) \defi {\left[ {\begin{array}{*{20}{c}}
{{{{g}}_1}(t)}&{{{{g}}_2}(t)}& \cdots &{{{{g}}_{{N_a}}}(t)}
\end{array}} \right]^T},
\end{align}
with the effective scalar channel for port $n$ defined as
\begin{align}
g_n(t) \defi {{\bf{f}}_n} \otimes {{\bf{s}}_n} {\underline{\bf{h}}}(t).
\end{align}

Then, the UE feeds back the quantized scalars $\hat{g}_n(t), n = 1, \cdots, N_a$ to the base station. Finally the DL channel is reconstructed as
\begin{align}\label{Eq:HhatDt-PCRE}
{{\bf{\widehat H}}}(t) = \sum\limits_{n = 1}^{{N_a}} {{{\hat g}_n}(t)} {\bf{s}}_n^H{\bf{f}}_n^*.
\end{align}

The proposed low-complexity alternative, i.e., PCR-E, is summarized in Algorithm \ref{Alg:PCR-E}.
\begin{algorithm}
\caption{PCR Codebook scheme with spatial/frequency eigen basis}
\begin{algorithmic}[1]\label{Alg:PCR-E}

\STATE{The base station obtains the precoders $\vs_n$ and $\vf_n, n = 1, \cdots, N_a$ by finding the positions of the $N_a$ greatest values in Eq. (\ref{Eq:calG});}

\STATE{The base station applies the precoders when transmitting CSI-RS;}

\STATE{For each antenna port, the UE estimates the effective channels in all $N_f$ subbands and sum them up to obtain $g_n(t)$;}

\STATE{The UE feeds back the quantized $g_n(t), n = 1, \cdots, N_a$;}

\STATE{The base station reconstructs the DL CSI with Eq. (\ref{Eq:HhatDt-PCRE});}

\end{algorithmic}
\end{algorithm}

Note that Algorithm \ref{Alg:PCR-E} assumes the eigenvectors ${\mU}^{{\rm{(S)}}}$ and ${\mU}^{{\rm{(F)}}}$ of the DL channel covariance matrices are known by the base station. This can be achieved with the UL-DL channel covariance transformation method in \cite{miretti2018}.

The complexity of this scheme at UE side is exactly the same as in the PCR scheme of Sec. \ref{sec:PCR}.
As for the base station, the main complexity comes from the searching of the spatial and frequency eigenmodes, which has the complexity order of $\mathcal{O}(N_c (N_t^2 N_f + N_t (N_f)^2)$. Nevertheless, the searching can be made faster by considering only the dominant eigenvectors, or even skipped at the cost of some mild performance loss.

\subsection{PCR codebook with DFT basis}\label{sec:PCR-D}
The two schemes mentioned above both needs the eigenvalue decomposition of channel covariance matrix. In this section, we propose an alternative named PCR codebook with DFT basis (PCR-D) that circumvents such an operation and thus has even lower complexity at the base station side. The key ingredient of this scheme is the exploitation of the channel reciprocity of the multipath angles and delays. Similar to PCR-E, the joint spatial-frequency precoder again consists of two components, the spatial precoder and the frequency precoder, which are DFT-based vectors instead of eigenvectors.

Denote a DFT matrix of size $K \times K$ as
\begin{align*}
\mbox{\small$\displaystyle
\mE(K) \triangleq \frac{1}{\sqrt{K}}\begin{bmatrix}
\omega^{0\cdot0} & \omega^{0\cdot1} & \cdots & \omega^{0(K-1)} \\
\omega^{1\cdot0} & \omega^{1\cdot1} & \cdots & \omega^{1(K-1)} \\
\vdots & \vdots & \ddots & \vdots \\
\omega^{(K-1)\cdot0} & \omega^{(K-1)\cdot 1}& \cdots & \omega^{(K-1)(K-1)} \end{bmatrix},$}
\end{align*}
where $\omega \triangleq e^{-2 \pi j / K}$.
According to the UPA antenna array structure, we construct a DFT-based spatial orthogonal basis as $\mS \defi \mE(N_h) \otimes \mE(N_v)$. Note that for the case of dual-polarized antennas, there are two options: we may either construct the spatial basis as
\begin{align}
\mS = \text{diag}\{\mE(N_h) \otimes \mE(N_v), \mE(N_h) \otimes \mE(N_v)\},
\end{align}
or process the two polarizations one by one. These two options are equivalent. Similarly, a frequency basis is defined as $\mF \defi \mE(N_f)$.

The spatial precoders are selected from the column vectors in $\mS$, while the frequency precoders are from $\mF$. The selection method is similar to PCR-E scheme, which relies on the projections of the UL channel samples ${{\bf{H}}^{{\rm{(U)}}}}({t_1}), {{\bf{H}}^{{\rm{(U)}}}}({t_2}), \cdots, {{\bf{H}}^{{\rm{(U)}}}}({t_{N_c}})$ on the spatial and frequency basis:
\begin{align}\label{Eq:calG_PCRE}
\bm{\mathcal{G}} \defi \sum\limits_{i = 1}^{{N_c}} \left( {\mS^H}{{\bf{H}}^{{\rm{(U)}}}}({t_i})\mF \right) \odot \left( {\mS^H}{{\bf{H}}^{{\rm{(U)}}}}({t_i})\mF \right)^*.
\end{align}
Denote the row index and column index of the $n$-th greatest value in $\bm{\mathcal{G}}$ as $r_n$ and $c_n$ respectively.
For the $n$-th CSI-RS port,
The spatial-domain precoder $\vs_n$ is
\begin{align}
\vs_n = \left(\mS_{r_n}\right)^H,
\end{align}
where $\mS_{r_n}$ is the $r_n$-th column of $\mS$.
The corresponding frequency-domain precoder $\vf_n$ is the $c_n$-th row of $\mF$. The joint spatial-frequency precoder is $\vf_n \otimes \vs_n$. The remaining operations are the same as in the PCR-E scheme and thus omitted. This scheme is summarized in Algorithm \ref{Alg:PCR-D}.
\begin{algorithm}
\caption{PCR Codebook scheme with DFT basis}
\begin{algorithmic}[1]\label{Alg:PCR-D}

\STATE{The base station obtains the precoders $\vs_n$ and $\vf_n, n = 1, \cdots, N_a$ by finding the positions of the $N_a$ greatest values in Eq. (\ref{Eq:calG_PCRE});}

\STATE{The base station applies the precoders when transmitting CSI-RS;}

\STATE{For each antenna port, the UE estimates the effective channels in all $N_f$ subbands and sum them up to obtain an estimate of $g_n(t)$;}

\STATE{The UE feeds back the quantized $g_n(t), n = 1, \cdots, N_a$;}

\STATE{The base station reconstructs the DL CSI by Eq. (\ref{Eq:HhatDt-PCRE});}

\end{algorithmic}
\end{algorithm}

\section{Numerical Results}\label{sec:numericalResult}
In this section, we show the simulation results of our proposed schemes, primarily under the practical model of 3GPP. The main parameters are listed in Table \ref{Tb:basicParas}. The configuration of the antennas, including the base station and the UE, is expressed by a tuple $(\underline{M},\underline{N},\underline{P})$, where $\underline{M}, \underline{N}, \underline{P}$ denote the number of rows, columns, and polarizations of the antenna elements respectively. We consider a typical parameter set of 5G with the center frequency at 3.5 GHz and the subcarrier spacing of 30 kHz. We adopt the CDL-A and CDL-D models in the evaluations. In CDL-A model, the total number of multipaths is 460, i.e., each UE has 23 multipath clusters and each cluster contains 20 paths. No Line-of-Sight (LoS) path exists. While in CDL-D channel model, there are totally 273 paths with one LoS component. The distributions of the paths in CDL-A model and CDL-D model are defined in Table 7.7.1-1 and Table 7.7.1-4 of \cite{3gpp:38.901} respectively.

The UEs have $N_r = 2$ antennas in simulations, and for each UE antenna, a set of $N_a$ coefficients are fed back to the BS. Therefore, the number of feedback coefficients is $2 N_a$ for the schemes of PCR, PCR-E, and PCR-D. The maximum number of streams for each UE is two.

\begin{table}[h]
\caption{Basic simulation parameters}\label{Tb:basicParas}
\centering
\begin{tabular}{|p{2.1cm}|p{5.5cm}|}
  \hline
  Scenario          & 3D Urban Macro (3D UMa) \\
  \hline
  DL center frequency           & 3.5 GHz \\
  \hline
  UL center frequency           & 3.4 GHz \\
    \hline
  Subcarrier spacing         & 30 kHz \\
    \hline
  Bandwidth  &  20 MHz (51 RBs)\\
    \hline
  Number of UEs  &  8\\
    \hline
  BS antenna configuration    & $(\underline{M}, \underline{N}, \underline{P}) = (2,8,2)/(4,8,2)$, $(D_h, D_v) = (0.5, 0.8)\lambda$, the polarization angles are $\pm 45^\circ$ \\
    \hline
  UE antenna configuration    & $(\underline{M}, \underline{N}, \underline{P}) = (1,1,2,1,1)$,  the polarization angles are $0^\circ$ and $90^\circ$\\
    \hline
  Channel model     & CDL-A, CDL-D \\
    \hline
  Delay spread       & 300 ns \\
    \hline
  DL precoder       & EZF \\
    \hline
  UE receiver       & MMSE-IRC \\
    \hline
  Number of paths       & 460, 271 \\
    \hline
\end{tabular}
\end{table}

We first visualize the effects of the dimension reduction of the three proposed methods in Fig. \ref{fig:dimensionReduction}.
\begin{figure}[htb]
	\centering
	\subfigure[]{\includegraphics[width=0.45\linewidth]{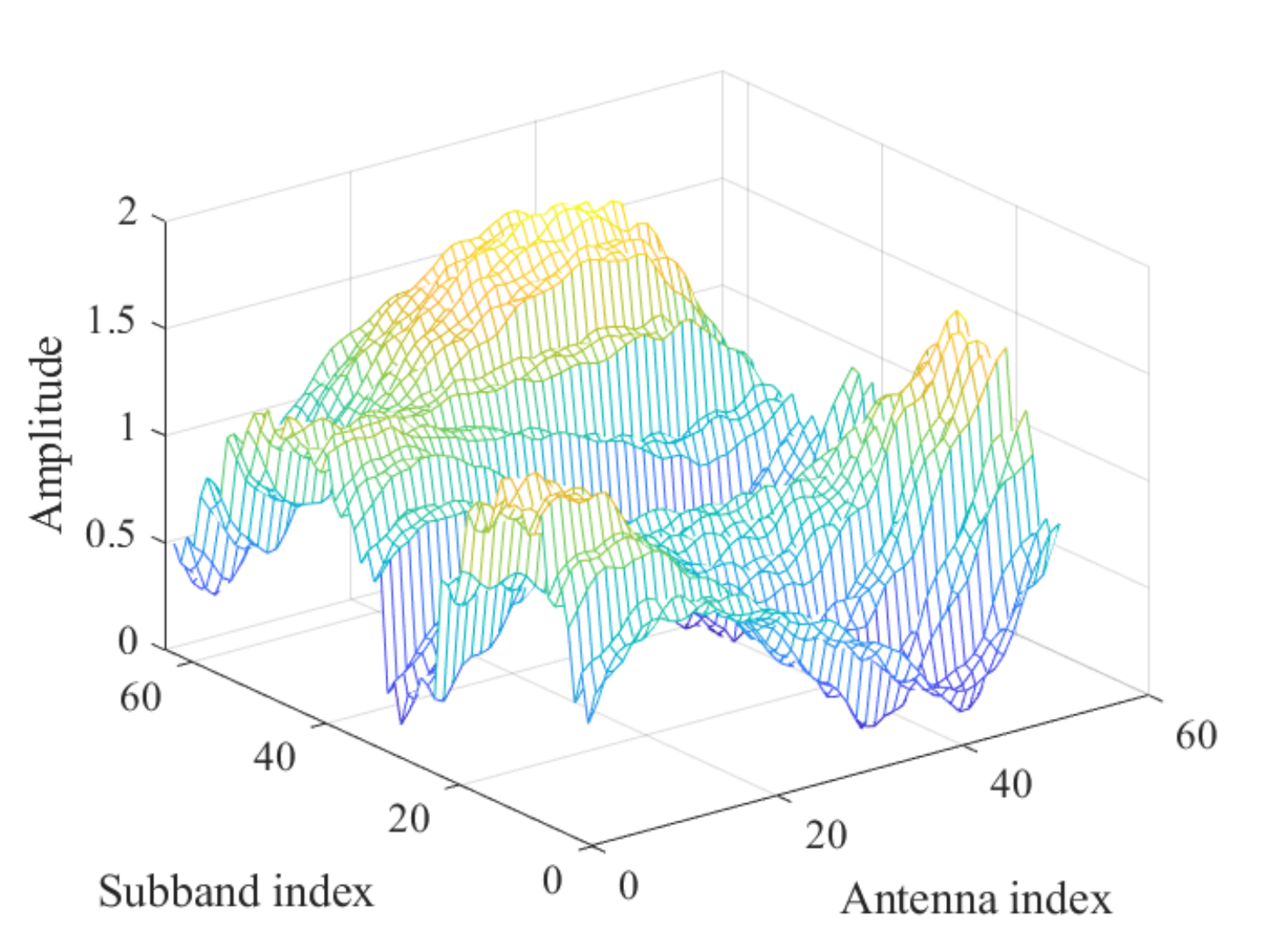}}%
	 \subfigure[]{\includegraphics[width=0.45\linewidth]{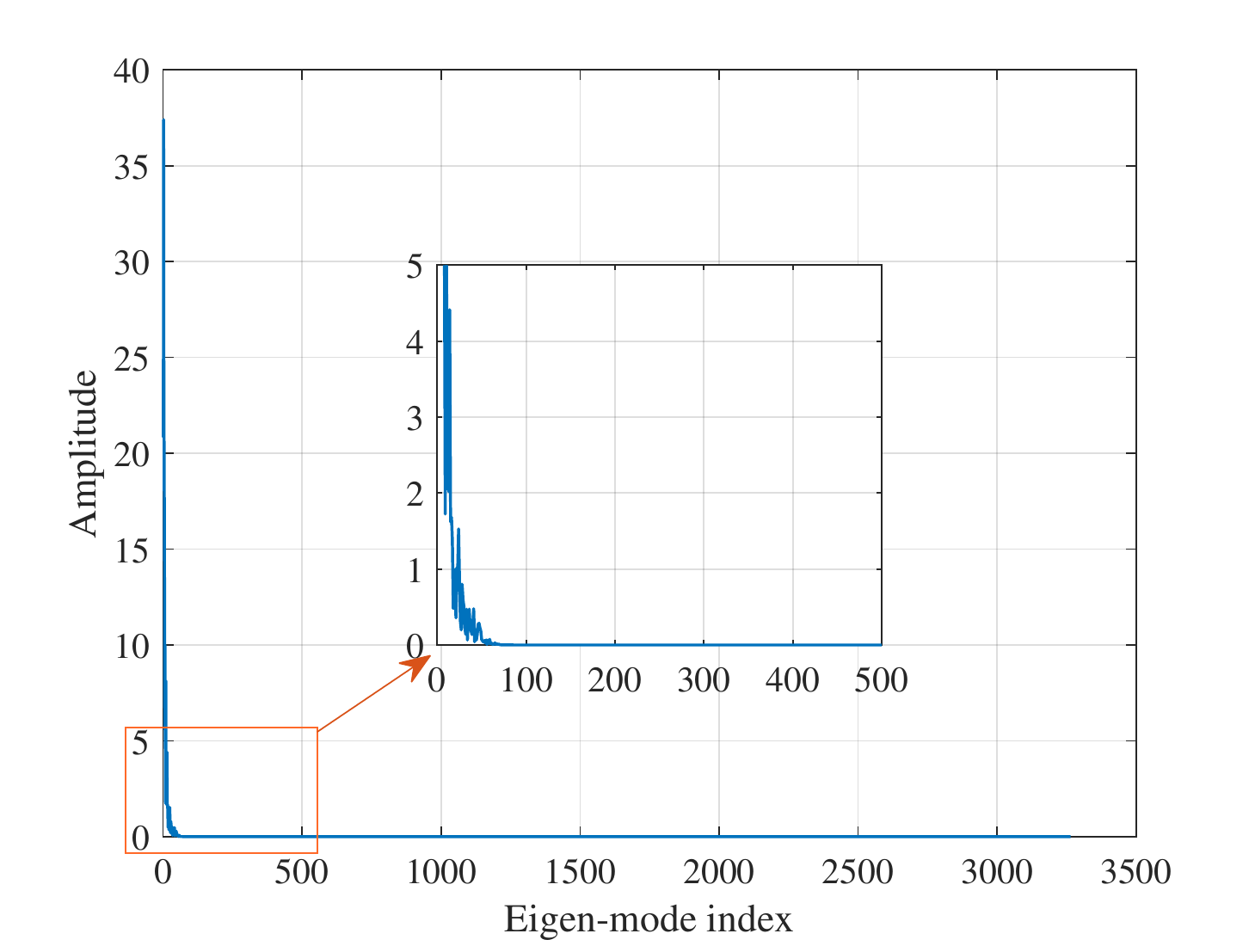}}%
	\hfill
	 \subfigure[]{\includegraphics[width=0.45\linewidth]{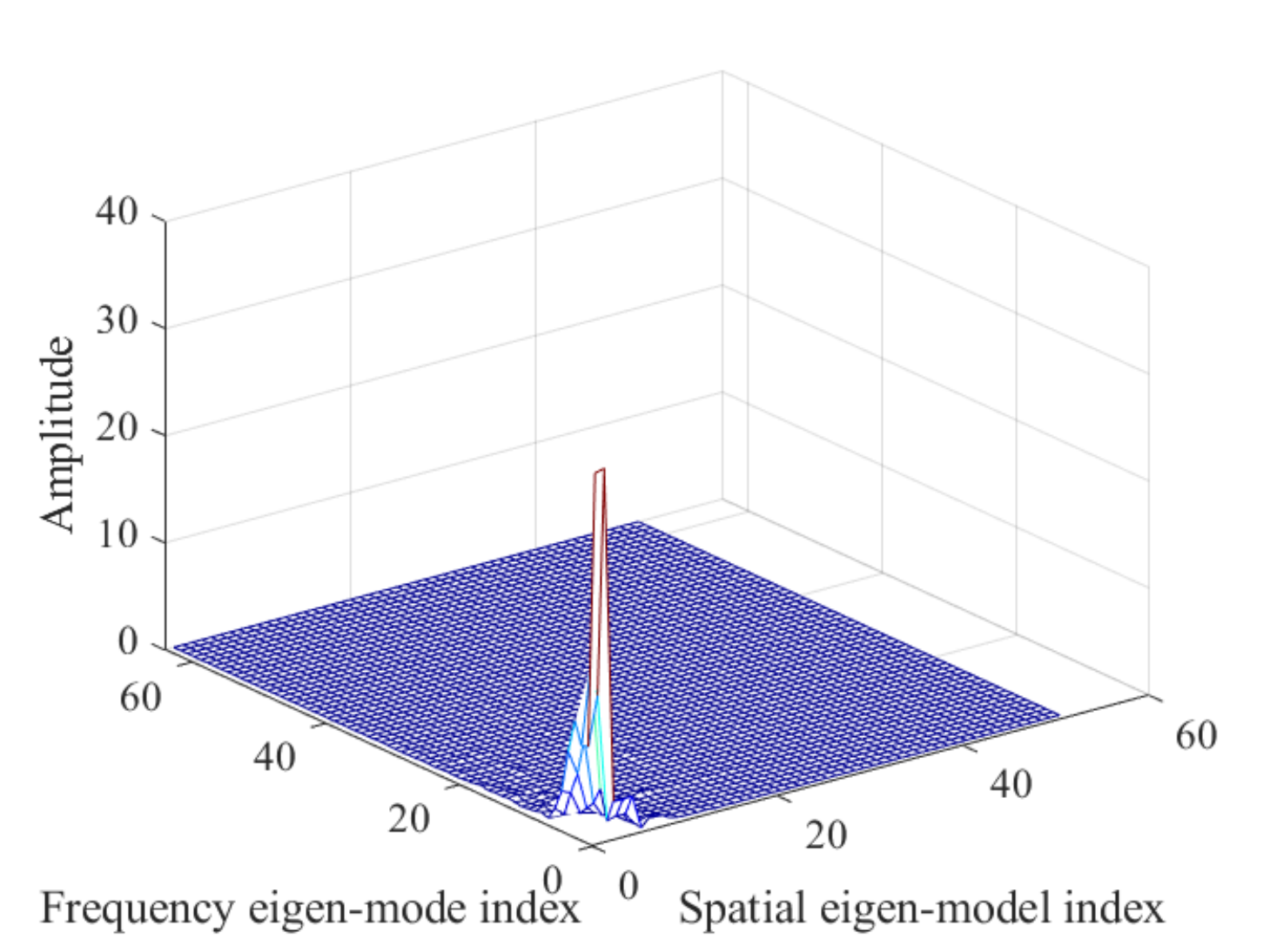}}%
	 \subfigure[]{\includegraphics[width=0.45\linewidth]{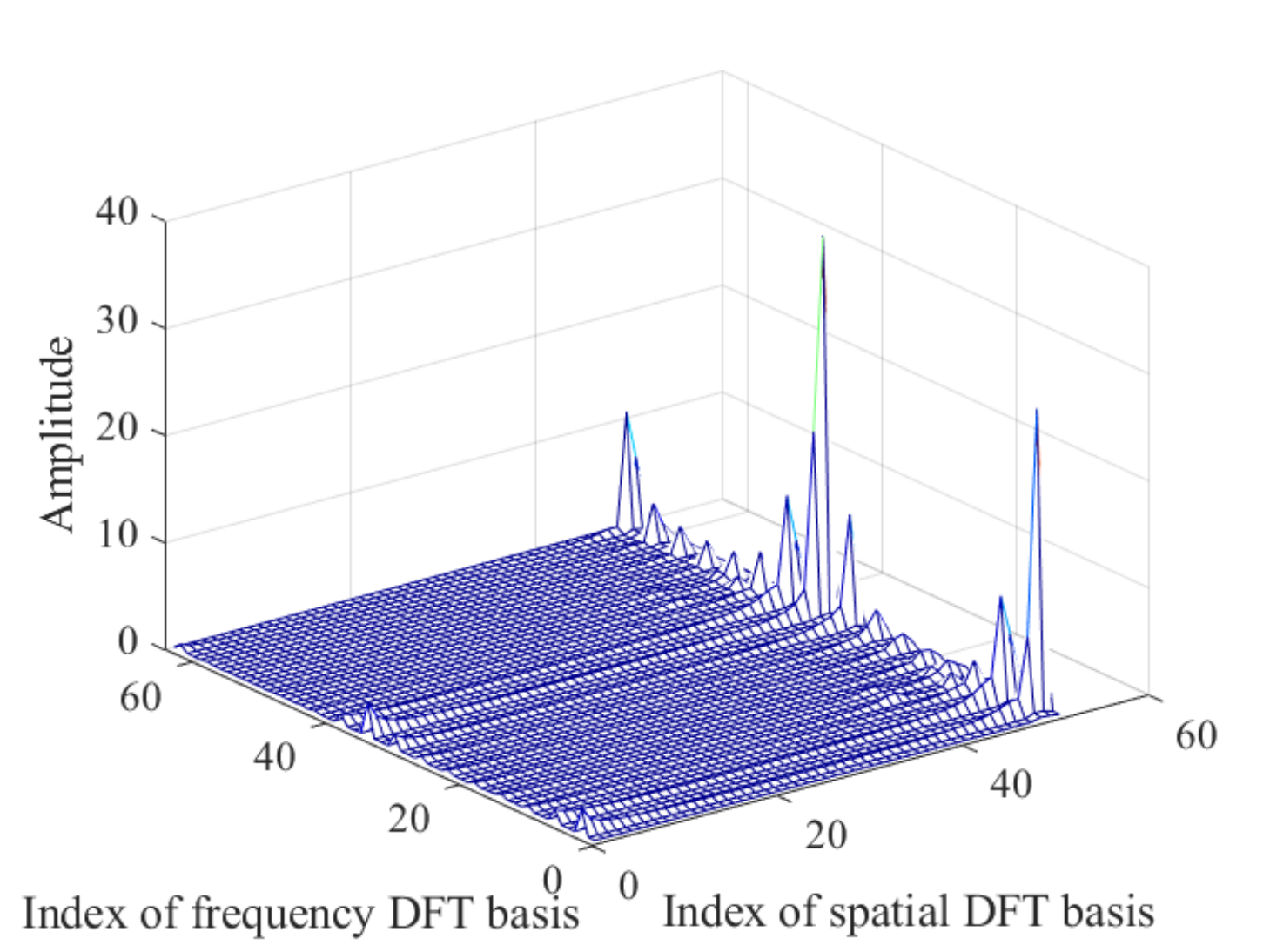}}%
	\caption{The element-wise amplitude of the wideband channel matrix and its projections (a) The channel matrix; (b) the projections onto the bases of the PCR scheme; (c) the projections onto the bases of the PCR-E scheme; (d) the projections onto the bases of the PCR-D scheme.}
	\label{fig:dimensionReduction}
\end{figure}
One snapshot of the wideband channel $\mH$ between the base station and UE is taken as an example. The antenna configuration at the base station is $(\underline{M},\underline{N},\underline{P}) = (4, 8, 2)$, and the number of subband is 51, as shown in Table \ref{Tb:basicParas}.
Fig. \ref{fig:dimensionReduction} (a) shows the amplitude of the  wideband channel coefficient of $\mH \in \mathbb{C}^{N_t \times N_f}$ where $N_t = 64, N_f = 51$. Fig. \ref{fig:dimensionReduction} (b) is the amplitude of the projected channel onto the orthogonal basis Eq. (\ref{Eq:EigsJ}) of the proposed PCR scheme in Sec. \ref{sec:PCR}. We may observe that only a very small portion of the projected values are significant. Such a high sparsity allows us to reconstructed the full channel matrix with only a small number of non-negligible scalar coefficients.
Fig. \ref{fig:dimensionReduction} (c) is the amplitude of the projection onto the basis of Eq. (\ref{Eq:EigsS}) and Eq. (\ref{Eq:EigsF}) as in the low-complexity alternative PCR-E scheme. Fig. \ref{fig:dimensionReduction} (d) is amplitude when projected on to the 2D-DFT basis as in the PCR-D scheme. From (c) and (d) we also observe that the projected channel matrix is sparse, and the eigen-mode based projection is better than the DFT-based projection in terms of the sparse representation of the channel matrix.

\begin{figure}[h]
  \centering
  \includegraphics[width=3.5in]{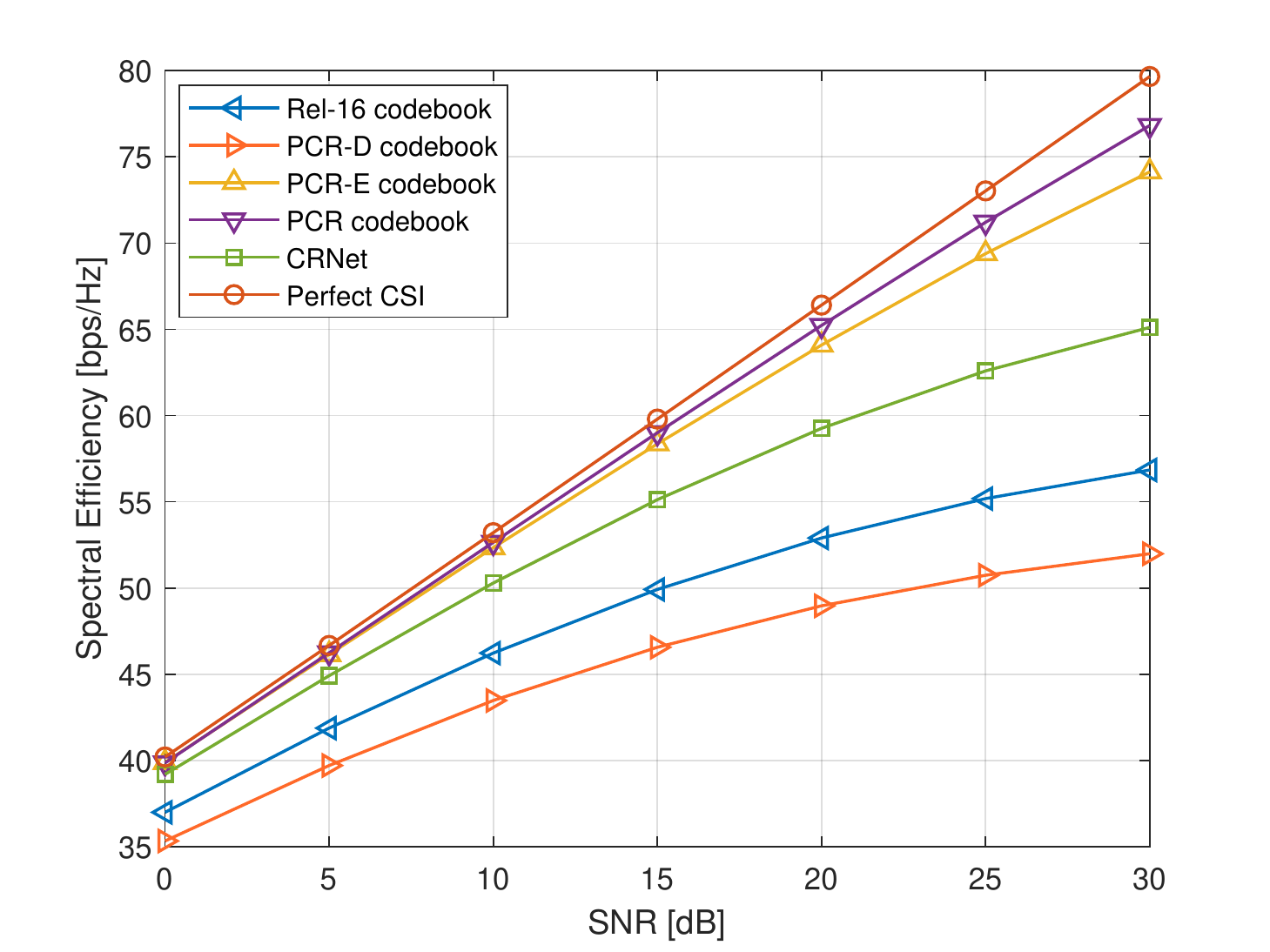}\\
  \caption{The spectral efficiency vs. SNR, $N_t = 32$, $N_a = 32$, CDL-A model.} \label{fig:32T2R_SE_vs_SNR}
\end{figure}

\begin{figure}[h]
  \centering
  \includegraphics[width=3.5in]{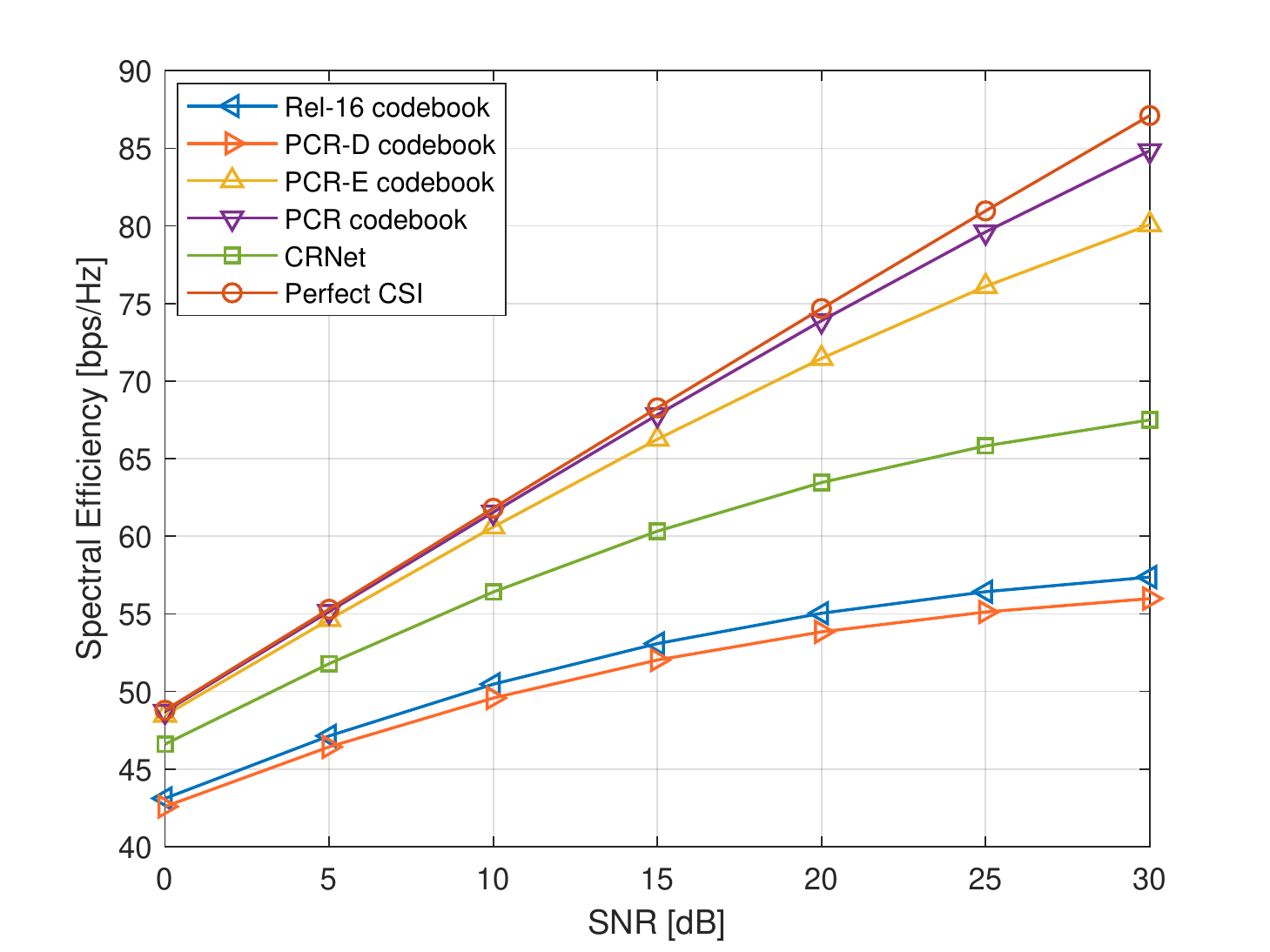}\\
  \caption{The spectral efficiency vs. SNR, $N_t = 64$, $N_a = 32$, CDL-A model.} \label{fig:64T2R_SE_vs_SNR}
\end{figure}

Fig. \ref{fig:32T2R_SE_vs_SNR} and Fig. \ref{fig:64T2R_SE_vs_SNR} shows the spectral efficiency performance of the proposed codebook schemes for the settings that the base station is equipped with 32 antennas and 64 antennas, respectively. The spectral efficiency is computed as $\sum_{n = 1}^{{N_u}} \log_2 (1 + \text{SINR}_n)$ averaged over time and frequency, where $\text{SINR}_n$ is the Signal and Interference to Noise Ratio for the $n$-th UE, and $N_u = 8$ is the total number of the active UEs. The number of coefficients to feed back is $2 N_a = 64$ for each of the schemes proposed by this paper. The benchmark of Enhanced Type II codebook in Rel-16 also requires $2 N_a$ feedback coefficients, i.e., two sets of $N_a$ 2D DFT projections corresponding to the two UE antennas respectively. However, the Rel-16 codebook requires the feedback of indices apart from the $2 N_a$ scalar coefficients, to notify the positions of the selected $N_a$ projections of the channel matrix on the 2D-DFT basis. The curve labeled with ``CRNet" is the Deep Learning-based CSI compression and reconstruction scheme proposed by \cite{Lu2020}. The CRNet consists of an encoder and decoder, where the encoder performs compression at the UE and the decoder reconstructs the original CSI at the BS. The training, validation, and test datasets contain 80000, 20000 and 20000 wideband channel samples respectively. The number of feedback coefficients for CRNet is also 64.
We may notice that our proposed PCR scheme is very close to the ideal case that perfect CSI is known at the base station. It is interesting to note that the wideband massive MIMO channel, although very large in size, can be effectively represented by only $2 N_a = 64$ scalar coefficients using our PCR codebook scheme. The low-complexity alternative of PCR-E has some mild performance loss compared to the PCR scheme. Nevertheless it has significant gains over the Enhanced Type II Codebook in the latest Rel-16 standard of 5G. The reason is that there are more power leakages under DFT basis than under eigen-basis, which is demonstrated in Fig. \ref{fig:dimensionReduction}. We may also observe that the low-complexity alternative of PCR-D scheme has lower performance than the Rel-16 codebook of 5G when the numbers of feedback coefficients are equal. However, PCR-D does not need to feedback indices, and the total amount of overhead is thus smaller. The complexity of PCR-D is also much lower.
In practice, there is always a trade-off between complexity and performance, and a substantial complexity reduction can be achieved at the cost of feeding back several more scalar coefficients.

\begin{table}[h!]
  \begin{center}
    \caption{Comparisons of the schemes}\label{Tb:SchemesComparison}
    \label{tab:table1}
    \begin{tabular}{|c|c|c|c|c|} 
      \hline
      {\textbf{Schemes}} & {\textbf{PCR} }/{\textbf{PCR-E}} & {\textbf{PCR-D}} & {\textbf{Rel-16}}\\
      \hline
      UE complexity & $\mathcal{O}(N_x N_f N_a)$  & $\mathcal{O}(N_x N_f N_a)$ & {\makecell[c]{$\mathcal{O}(N_x N_f N_a) + \mathcal{O}$ \\ $(N_t N_f \log (N_t N_f))$}} \\
      \hline
      Overhead & 2 $N_a$ scalars & 2 $N_a$ scalars & {\makecell[c]{2 $N_a$ scalars + \\$N_a$ indices}}\\
      \hline
      \# of ports & $N_a$ & $N_a$ & $N_t$\\
      \hline
      Generality & high  & UPA only & UPA only\\
      \hline
    \end{tabular}
  \end{center}
\end{table}
Table \ref{Tb:SchemesComparison} is the comparison of the proposed schemes and the Enhanced Type II codebook of the Rel-16 in terms of UE complexity, feedback overhead, required number of antenna ports, and generality. Since the Enhanced Type II codebook adopts non-precoded CSI-RS, the number of antenna ports is equal to the number of BS antennas $N_t$. The PCR-D and Enhanced Type II codebook both rely on 2D DFT for CSI compression, which limits their applicability to UPA topology of BS antennas. However, the PCR and PCR-E codebooks are more general and applicable to other BS antenna topologies as well.
\begin{figure}[h]
  \centering
  \includegraphics[width=3.5in]{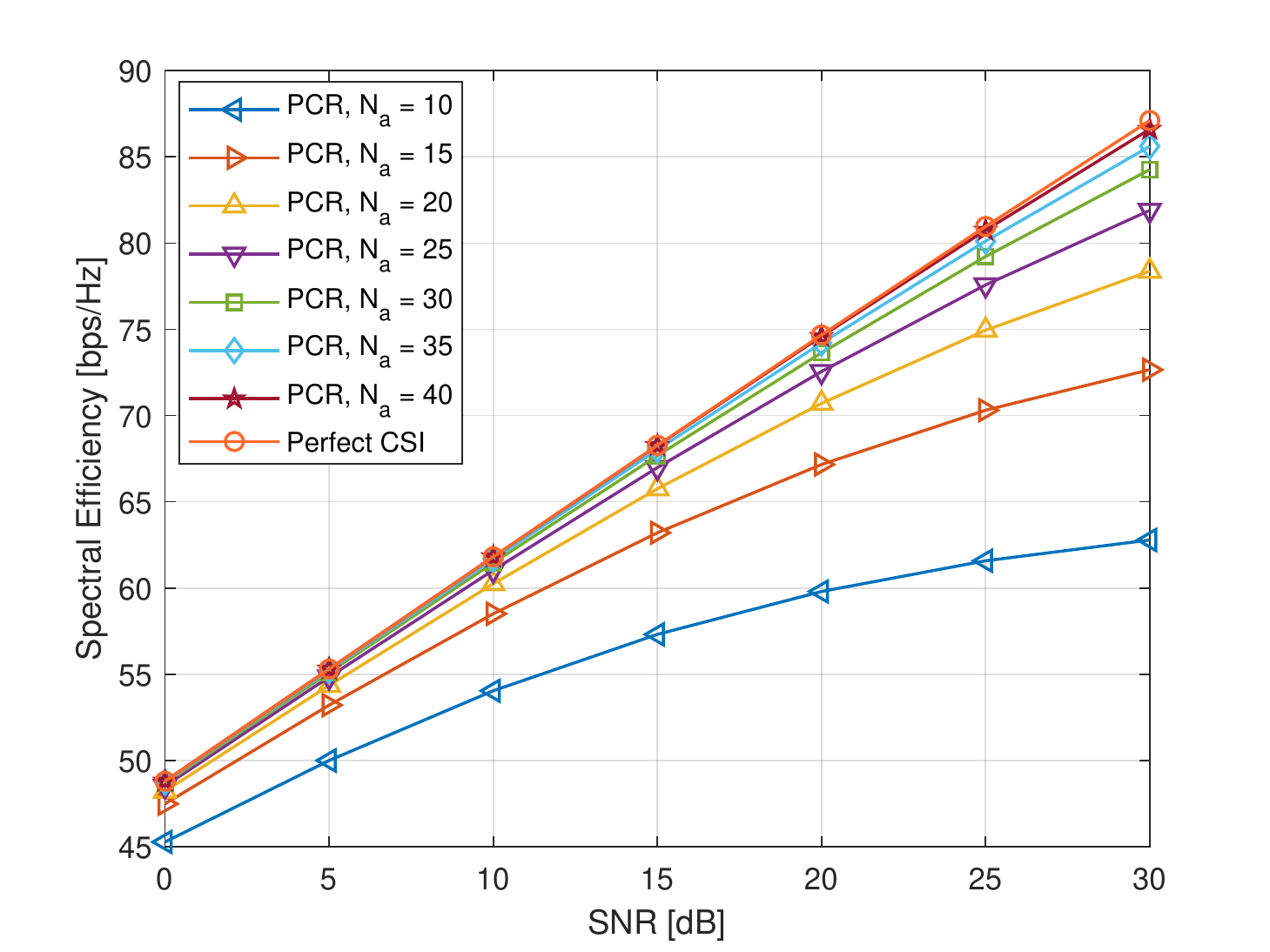}\\
  \caption{The spectral efficiency vs. SNR, PCR with different feedback overhead, $N_t = 64$, CDL-A model.} \label{fig:64T2R_PCR_with_different_Nas}
\end{figure}

Then we show the impact of the number of antenna ports $N_a$ on the PCR scheme in Fig. \ref{fig:64T2R_PCR_with_different_Nas}. When the value of $N_a$ increases from 10 to 40, the performance grows monotonically and eventually saturates near the ideal case of perfect CSI.
\begin{figure}[h]
  \centering
  \includegraphics[width=3.5in]{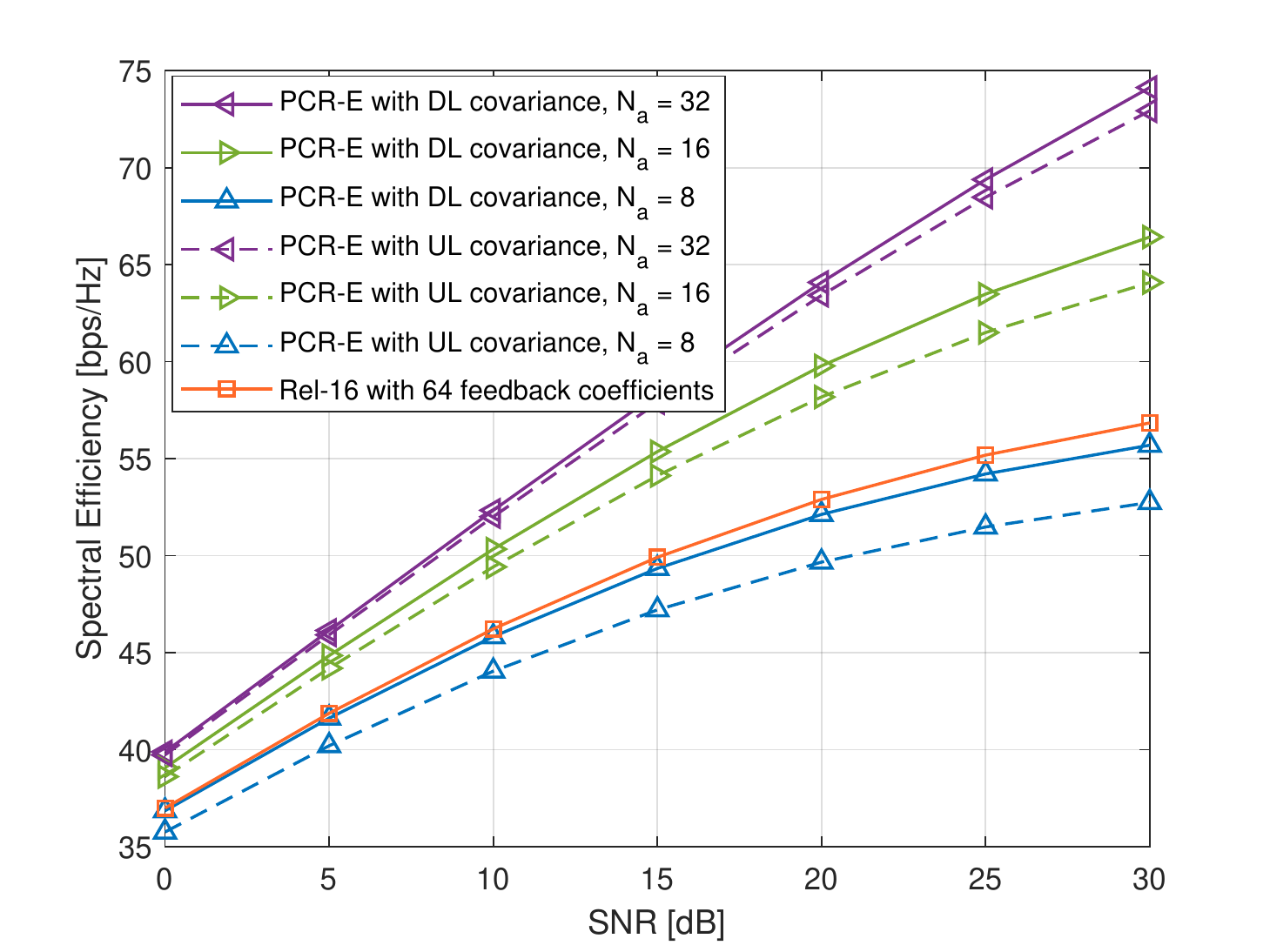}\\
  \caption{The spectral efficiency vs. SNR, PCR-E with different feedback overhead, $N_t = 32$, CDL-A model.} \label{fig:32T2R_PCRE_with_different_Nas}
\end{figure}

Now we show in Fig. \ref{fig:32T2R_PCRE_with_different_Nas} the performance of the PCR-E scheme with different feedback overheads, i.e., $N_a = 32, 16, 8$. In order to demonstrate the robustness of our scheme to inaccurate channel covariance matrices, e.g., $\mR^\text{(S)}$ and $\mR^\text{(F)}$, we also show the cases where the DL channel covariance matrices are replaced by the UL ones. In other words, the curves labeled with ``PCR-E with UL covariance" means the eigenvectors of the UL channel covariance matrices are used when computing the precoders. The transformations of the UL covariance to the DL version are not performed, which further reduces the computational complexity for the base station. It is interesting to note that even with such non-ideal channel covariance matrices, the degradation of the performance is tolerable. Moreover, the Rel-16 codebook with 64 feedback coefficients are also plotted as a benchmark. One may observe that the PCR-E scheme with much fewer feedback coefficients, i.e., 16 ($N_a = 8$), is comparable with this benchmark.
\begin{figure}[h]
  \centering
  \includegraphics[width=3.5in]{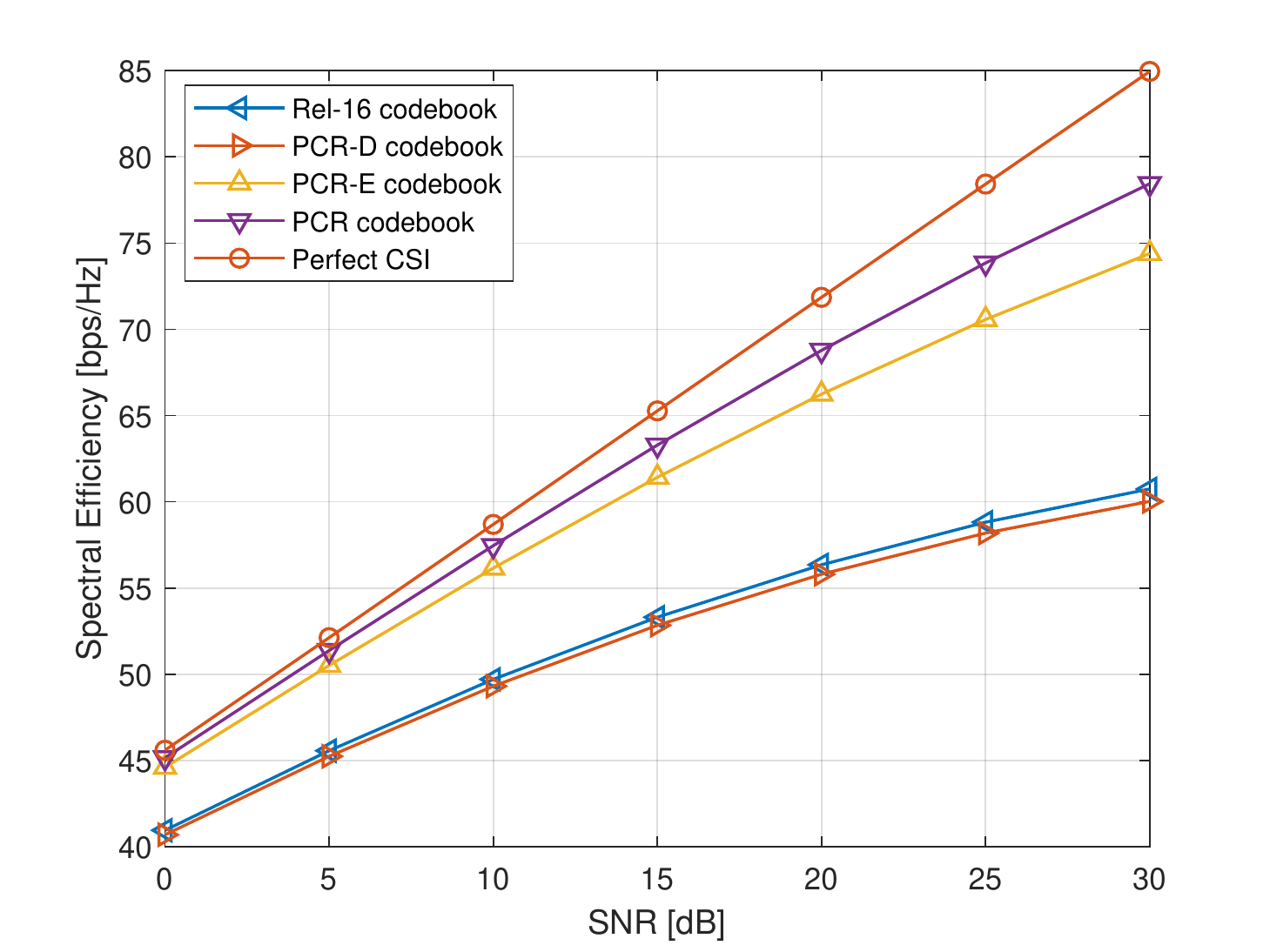}\\
  \caption{The spectral efficiency vs. SNR, $N_t = 64$, $N_a = 20$, CDL-D model.} \label{fig:64T2R_SE_vs_SNR_CDLD}
\end{figure}

Then, the CDL-D model is considered and the performance curves are shown in Fig. \ref{fig:64T2R_SE_vs_SNR_CDLD}. Since there are less paths and one of them is an LoS component in CDL-D, the channel matrix has more sparsity in terms of angles and delays compared to CDL-A. It leads to higher correlation between antennas and subcarriers, and therefore makes it possible to reconstruct the DL channel with less feedback coefficients. As we may observe in the figure, the PCR scheme with $N_a = 20$ is already very close to the perfect CSI setting.

\begin{figure}[h]
  \centering
  \includegraphics[width=3.5in]{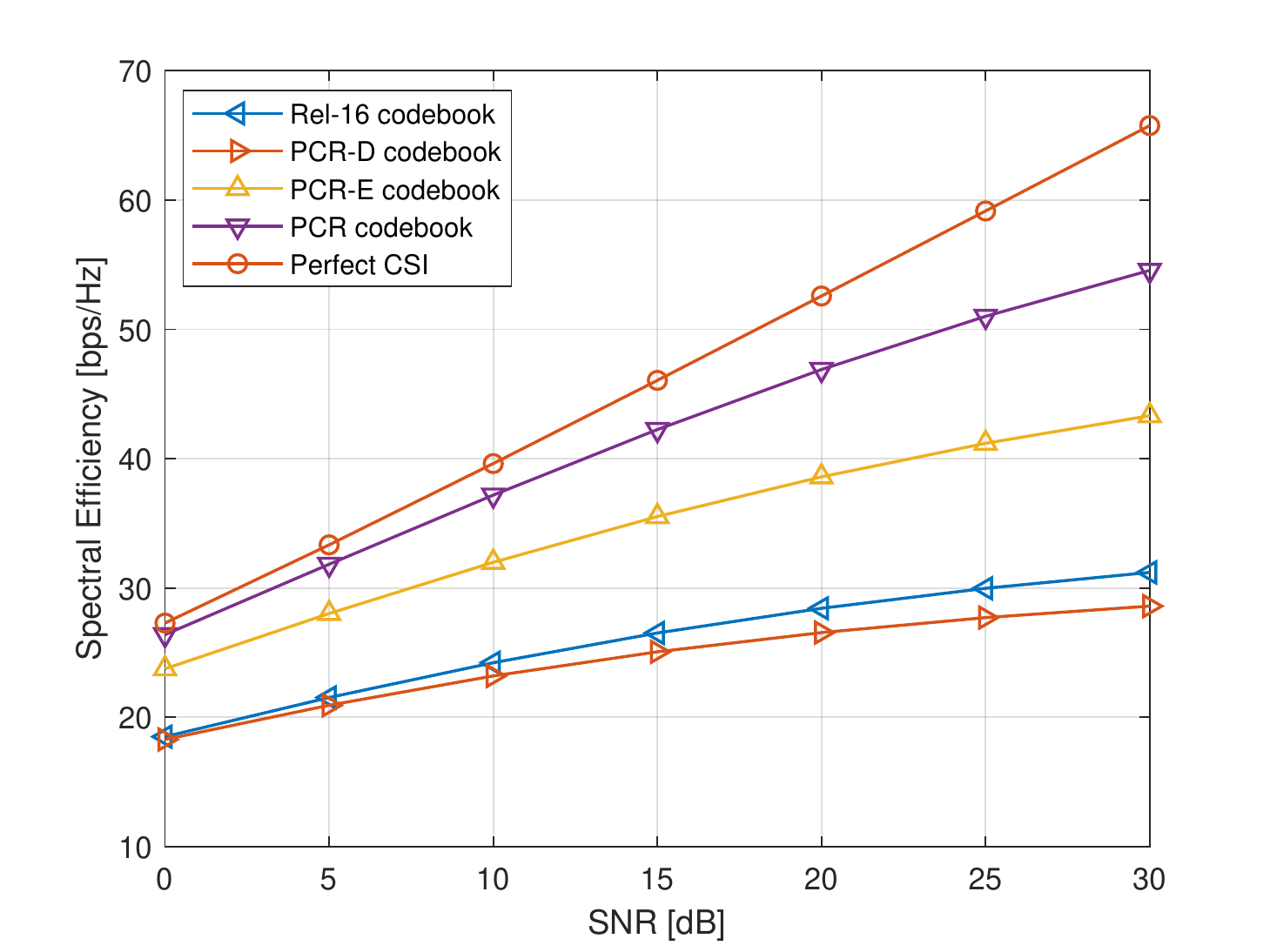}\\
  \caption{The spectral efficiency vs. SNR, $N_t = 32$, $N_a = 32$, 8 UEs, COST-2100 model.} \label{fig:32T1R_SE_vs_SNR_cost2100}
\end{figure}

Finally, the proposed schemes are evaluated under the COST-2100 channel model \cite{Cost2100} in Fig. \ref{fig:32T1R_SE_vs_SNR_cost2100}. The BS antenna topology is ULA with 32 antenna elements, and all 8 UEs have single antenna. The bandwidth is 20 MHz, which contains 51 RBs. The total number of feedback coefficients is 32 for all the schemes. One may observe that the proposed PCR and PCR-E codebooks outperforms the Enhanced Type II codebook under COST-2100 channel model.
\section{Conclusions}\label{sec:conclusion}
In this paper we dealt with the challenge of CSI acquisition in FDD massive MIMO. We first derive the ranks of the channel covariance matrices of the wideband large-scale MIMO channel for a given angle and delay distribution. The closed-form expression of the rank indicates that a low-rankness property is always valid for a UPA with half-wavelength or smaller antenna spacing, regardless of the richness of the scattering environment. We then proposed a PCR codebook scheme which exploits such a property and the partial reciprocity of UL/DL channels.
It outperforms the latest 5G codebook of Rel-16 at performance, UE complexity, feedback overhead, received SNR of CSI-RS, and generality.
We also proposed two alternatives named PCR-E and PCR-D codebook schemes, aiming at further reducing the complexity at the BS side at the cost of mild performance loss. Simulation results showed that the proposed methods are very close to the ideal case of perfect CSI even with only tens of scalar feedback coefficients for the whole wideband massive MIMO channel.

\appendix
\subsection{Proof of Theorem \ref{theoSpatialRank}:}\label{proof:theoSpatialRank}
\begin{proof}
In order to prove Theorem \ref{theoSpatialRank}, we need to first prove that non-overlapping angular supports lead to additive dimensionality of the signal subspace of the spatial channel covariance matrix.
\begin{lemma}\label{lemma:additiveRank}
Consider two angular supports $\zeta_1$ and $\zeta_2$. The boundaries are defined as
\begin{align}
\zeta_1 = \left\{ (\theta, \varphi) | \theta \in [\theta_1^\text{min}, \theta_1^\text{max}], \varphi \in [\varphi_1^\text{min}(\theta), \varphi_1^\text{max}(\theta)]\right\}, \\
\zeta_2 = \left\{ (\theta, \varphi) | \theta \in [\theta_2^\text{min}, \theta_2^\text{max}], \varphi \in [\varphi_2^\text{min}(\theta), \varphi_2^\text{max}(\theta)]\right\}.
\end{align}
The spatial covariance matrices corresponding to the angular supports are denoted as $\mR_1$ and $\mR_2$ respectively.
Define the linear space spanned by the 2-D steering vectors with corresponding angles inside the angular support $\zeta$ as
\begin{align}\label{Eq:Bcal}
\mathcal{B} \defi \text{span} \left\{ {\frac{{{\bf{a}}(\theta ,\varphi )}}{\sqrt{{{N_h}{N_v}}}},\left( {\theta ,\varphi } \right) \in \zeta } \right\}.
\end{align}

If the two angular supports are non-overlapping, i.e., $\zeta_1 \cap \zeta_2 = \emptyset$, then the null space of $\mR_1$ includes the linear space $\mathcal{B}_2$ when $N_h$ and $N_v$ are large:
\begin{align}
\text{null}({\mR_1}) \supset \mathcal{B}_2, \text{ as } N_h, N_v \rightarrow \infty.
\end{align}
\end{lemma}
\begin{proof}
\quad \emph{Proof:} We denote the joint angular power spectrum of the multipath as ${p_i}(\theta ,\varphi )$, which is finite and uniformly non-zero inside the angular support $\zeta_i$, for $i = 1, 2$. Without loss of generality, we normalize the total power to 1, such that
\begin{align}
{\left( {\int_{\theta _i^{\min }}^{\theta _i^{\max }} {\int_{\varphi _i^{\min }(\theta )}^{\varphi _i^{\max }(\theta )} {p_i}(\theta ,\varphi ) {d\varphi d\theta } } } \right) = 1}.
\end{align}
For ease of exposition, we use ${p_i}$ to represent ${p_i}(\theta ,\varphi )$ in the proof below.

The covariance matrix corresponding to the angular support $\zeta_i, i = 1, 2$ is written as
\begin{align}\label{Eq:RiE}
{\mR_i} &= \mathbb{E}\{ {\bf{a}}({\theta _i},{\varphi _i}){\left( {{\bf{a}}({\theta _i},{\varphi _i})} \right)^H}\} \\
 &= \int_{\theta _i^{\min }}^{\theta _i^{\max }} {\int_{\varphi _i^{\min }(\theta )}^{\varphi _i^{\max }(\theta )} {{\bf{a}}(\theta ,\varphi ){{\left( {{\bf{a}}(\theta ,\varphi )} \right)}^H}{p_i}d\varphi d\theta } }.
\end{align}
Taking an arbitrary angle $(\theta_2, \varphi_2) \in \zeta_2$, then we have
\begin{small}
\begin{align}
&\frac{{{{\left( {{\bf{a}}({\theta _2},{\varphi _2})} \right)}^H}}}{\sqrt{{N_h}{N_v}}}{{\bf{R}}_1}\frac{{{\bf{a}}({\theta _2},{\varphi _2})}}{\sqrt{{N_h}{N_v}}} \nonumber \\
& = \frac{{1}}{{{N_h}{N_v}}} \iint {p_1}{{{{\left( {{\bf{a}}({\theta _2},{\varphi _2})} \right)}^H}}{\bf{a}}(\theta ,\varphi ){{\left( {{\bf{a}}(\theta ,\varphi )} \right)}^H}{{\bf{a}}({\theta _2},{\varphi _2})}d\varphi d\theta } \nonumber \\
& = \frac{{1}}{{{N_h}{N_v}}} \iint {p_1} {\left( {{{\left( {{{\bf{a}}_h}({\theta _2},{\varphi _2}) \otimes {{\bf{a}}_v}({\theta _2})} \right)}^H}\left( {{{\bf{a}}_h}({\theta},{\varphi}) \otimes {{\bf{a}}_v}({\theta})} \right)} \right)^2}d\varphi d\theta \nonumber \\
& = \frac{{1}}{{{N_h}{N_v}}} \iint {p_1} {\left( {{{\left( {{{\bf{a}}_h}({\theta _2},{\varphi _2})} \right)}^{^H}}{{\bf{a}}_h}(\theta ,\varphi ){{\left( {{{\bf{a}}_v}({\theta _2})} \right)}^H}{{\bf{a}}_v}(\theta )} \right)^2}d\varphi d\theta \nonumber \\
& = \frac{{1}}{{{N_h}{N_v}}} \iint {p_1} \left( {\left( {\sum\limits_{n = 0}^{{N_h-1}} {{e^{\frac{{j2\pi n {D_h}}}{\lambda }\left( {\sin (\theta )\sin (\varphi ) - \sin ({\theta _2})\sin ({\varphi _2})} \right)}}} } \right.} \right) \nonumber \\
& \qquad \qquad {\left. {\left( {\sum\limits_{n = 0}^{{N_v-1}} {{e^{\frac{{j2\pi n{D_v}}}{\lambda }\left( {\cos (\theta ) - \cos ({\theta _2})} \right)}}} } \right)} \right)^2}d\varphi d\theta. \label{Eq:sumAll}
\end{align}
\end{small}
Since the angle pair $({\theta _2},{\varphi _2})$ is outside the boundary of the double integral $\zeta_1$, we readily have that the terms ${\sin (\theta )\sin (\varphi ) - \sin ({\theta _2})\sin ({\varphi _2})}$ and ${\left( {\cos (\theta ) - \cos ({\theta _2})} \right)}$  never take the value zero at the same time. Thus the expression (\ref{Eq:sumAll}) converges to zero as $N_h$ and $N_v$ go to infinity. Therefore Lemma \ref{lemma:additiveRank} is proved.
\end{proof}

Note that the linear space $\mathcal{B}_i$ defined in Eq. (\ref{Eq:Bcal}) is equivalent to the signal subspace of the spatial channel covariance matrix $\mR_i, i = 1, 2$, i.e.,
\begin{align}
\mathcal{B}_i = \text{span}\left\{ {\vu_n^{(i)}:n = 1, \cdots ,{r_i}} \right\},
\end{align}
where $\vu_n^{(i)}$ is the eigenvector of $\mR_i$ corresponding to its $n$-th greatest eigenvalue and $r_i$ denotes the rank of $\mR_i$.
According to Lemma \ref{lemma:additiveRank}, the signal subspaces of $\mR_1$ and $\mR_2$ are asymptotically orthogonal when $N_h$ and $N_v$ are large:
\begin{align}\label{Eq:b1b2}
\mathcal{B}_1 \bot \mathcal{B}_2, \text{ as } N_h, N_v \rightarrow \infty.
\end{align}
We define the angular support $\zeta_u$ as the union of $\zeta_1$ and $\zeta_2$:
\begin{align}
\zeta_u \defi \zeta_1 \cup \zeta_2,
\end{align}
and the spatial channel covariance matrix $\mR_u$ corresponding to the multipath angular support $\zeta_u$. Then the rank of the spatial channel covariance $\mR_u$, denoted by $r_u$, satisfies
\begin{align}
r_u = r_1  + r_2, \text{ as } N_h, N_v \rightarrow \infty.
\end{align}

The above observation implies that the rank of the spatial channel covariance matrix can be computed by summing up the contributions of all non-overlapping angular sub-supports that forms the entire range. More precisely, consider a spatial covariance $\mR_u$ of a certain user that has an angular support $\zeta_u$. Such a support is composed of $K$ non-overlapping sub-supports:
\begin{align}
\zeta_u &= \zeta_1 \cup \zeta_2 \cdots \cup \zeta_K, \\
\zeta_i & \cap \zeta_j = \emptyset, i, j, = 1, 2, \cdots, K, \forall i \neq j.
\end{align}
Then the rank of $\mR_u$ yields:
\begin{align}
r_u = \sum\limits_{i = 1}^K {{r_i}} , \text{ as } N_h, N_v \rightarrow \infty,
\end{align}
where ${r_i}$ is the rank of the spatial covariance matrix with angular support $\zeta_i$.
Therefore, finding the rank of the spatial covariance matrix is equivalent to computing the size of the angular support $\zeta_u$, which could be done by a double integral over the boundary of the angular support, that is
\begin{align}
& \mathop {\lim }\limits_{{N_h},{N_v} \to \infty } \frac{{{\rm{rank}}\{ {{\bf{R}}^{({\rm{S}})}}\} }}{{{N_h}{N_v}}} \label{Eq:doubleIntegral} \\
& = \int\limits_{\cos ({\theta ^{\max }})}^{\cos ({\theta ^{\min }})} {\int\limits_{\sin ({\varphi ^{\min }}(\theta ))}^{\sin ({\varphi ^{\max }}(\theta ))} {\frac{{{D_h}{D_v}}}{{{\lambda ^2}}}\sin (\theta )d\left( {\sin (\varphi )} \right)d\left( {\cos (\theta )} \right)} } \nonumber \\
& =  \frac{{{D_h}{D_v}}}{{{\lambda ^2}}}\int\limits_{{\theta ^{\min }}}^{{\theta ^{\max }}} {\int\limits_{{\varphi ^{\min }}(\theta )}^{{\varphi ^{\max }}(\theta )} {{{\left( {\sin (\theta )} \right)}^2}\cos (\varphi )d\varphi d\theta } } \nonumber \\
& = \frac{{{D_h}{D_v}}}{{{\lambda ^2}}}\int\limits_{{\theta ^{\min }}}^{{\theta ^{\max }}} {{{\sin }^2}(\theta )\left( {\sin \left( {{\varphi ^{\max }}(\theta )} \right) - \sin \left( {{\varphi ^{\min }}(\theta )} \right)} \right)d\theta }. \nonumber
\end{align}
Thus, Theorem \ref{theoSpatialRank} is proved.
\end{proof}
\subsection{Proof of Theorem \ref{theoJointRank}:}\label{proof:theoJointRank}
\begin{proof}
Based on the additive dimensionality of the signal subspaces proved in Lemma \ref{lemma:additiveRank}, we only derive the rank of the joint spatial-frequency channel covariance matrix $\mR_q^\text{(J)}$ with a certain sub-support $\eta_q$. In a way similar to prove Theorem \ref{theoSpatialRank}, we may compute the rank by a triple integral over the three-dimensional range of the sub-support.

\begin{align}
& \mathop {\lim }\limits_{{N_h},{N_v},N_f \to \infty } \frac{{{\rm{rank}}\{ {{\bf{R}}_q^{({\rm{J}})}}\} }}{{{N_h}{N_v}{N_f}}} = \\
& \frac{{{D_h}{D_v}\Delta f}}{{{\lambda ^2}}}\int\limits_{\cos ({\theta ^{\max }})}^{\cos ({\theta ^{\min }})} \int\limits_{\sin ({\varphi ^{\min }}(\theta ))}^{\sin ({\varphi ^{\max }}(\theta ))} \int\limits_{{\tau ^{\min }}(\theta ,\varphi )}^{{\tau ^{\max }}(\theta ,\varphi )} \sin (\theta ) \nonumber \\
& \qquad d\left( {\sin (\varphi )} \right)d\left( {\cos (\theta )} \right)   d\tau \nonumber \\
& = \frac{{{D_h}{D_v}\Delta f}}{{{\lambda ^2}}}\int\limits_{\theta _q^{\min }}^{\theta _q^{\max }} {\int\limits_{\varphi _q^{\min }(\theta )}^{\varphi _q^{\max }(\theta )} {\int\limits_{\tau _q^{\min }(\theta ,\varphi )}^{\tau _q^{\max }(\theta ,\varphi )} {{{\sin }^2}(\theta )\cos (\varphi )} } } d\varphi d\theta d\tau \nonumber \\
& = \frac{{{D_h}{D_v}\Delta f}}{{{\lambda ^2}}}\int\limits_{\theta _q^{\min }}^{\theta _q^{\max }} \int\limits_{\varphi _q^{\min }(\theta )}^{\varphi _q^{\max }(\theta )} {{\sin }^2}(\theta ) \cos (\varphi ) \nonumber \\
& \qquad \left( {\tau _q^{\max }(\theta ,\varphi ) - \tau _q^{\min }(\theta ,\varphi )} \right)d\varphi d\theta   \nonumber
\end{align}

Thus, Theorem \ref{theoJointRank} is proved.
\end{proof}

\subsection{Proof of Theorem \ref{theoFeedbackAmount}:}\label{proof:theoFeedbackAmount}
\begin{proof}
The operation of the joint spatial-frequency beamforming with $N_a$ dominant eigen-vectors of $\mR^{\text{(J)}}$ is equivalent to the following matrix multiplication
\begin{align}
\vc_s = (\mU_s^{\text{(J)}})^H {\underline{\bf{h}}},
\end{align}
where the columns of $\mU_s^{\text{(J)}}$ are the dominant eigenvectors corresponding to the $N_a$ largest eigenvalues of $\mR^{\text{(J)}}$:
\begin{align}
\mU_s^{\text{(J)}} = \left[ {\begin{array}{*{20}{c}}
{{\bf{u}}_1^{(J)}}&{{\bf{u}}_2^{(J)}}& \cdots &{{\bf{u}}_{{N_a}}^{(J)}}
\end{array}} \right],
\end{align}
where $\vc_s \in \mathbb{C}^{N_a \times 1}$ are the coefficients that UE feeds back to the base station. Thus the reconstructed CSI at the base station side is
\begin{align}
{\underline {\widehat {\bf{h}}}} = {\bf{U}}_s^{{\rm{(J)}}}{({\bf{U}}_s^{{\rm{(J)}}})^H}{\underline {\bf{h}}}.
\end{align}
The expression above is equivalent to an orthogonal project the vector ${\underline {\bf{h}} }$ to the column space of $\mU_s^{\text{(J)}}$. As proved in Theorem \ref{theoJointRank}, the rank of $\mR^{\text{(J)}}$, denoted by $r^{\text(J)}$, is no larger than $\rho^{{\rm{(J)}}} N_v N_h N_f$. We have that ${\underline {\bf{h}} }$ lives in the signal subspace of $\mR^{\text{(J)}}$:
\begin{align}
{\underline {\bf{h}}} \in \text{span}\left\{ {\vu_n^{(J)}:n = 1, \cdots ,{r^{\text{(J)}}}} \right\}.
\end{align}
As a result, when the number of feedback coefficients $N_a \geq {r^\text{(J)}}$,  we have
\begin{align}
\mathop {\lim }\limits_{{N_h},{N_v},{N_f} \to \infty } \frac{{\left\| {{{\underline {{\bf{\hat h}}} }} - {\underline {\bf{h}} }} \right\|_F^2}}{{\left\| {{{\underline {\bf{h}} }}} \right\|_F^2}} = 0,
\end{align}
which proves Theorem \ref{theoFeedbackAmount}.
\end{proof}

\bibliographystyle{IEEEtran}
\bibliography{bib/allCitations}
\begin{IEEEbiography}[{\includegraphics[width=1in,height=1.25in,clip,keepaspectratio]{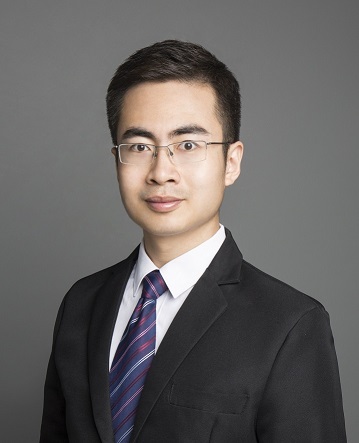}}]
{Haifan Yin} received the Ph.D. degree from T\'el\'ecom ParisTech in 2015. He received the B.Sc. degree in Electrical and Electronic Engineering and the M.Sc. degree in Electronics and Information Engineering from Huazhong University of Science and Technology, Wuhan, China, in 2009 and 2012 respectively. From 2009 to 2011, he has been with Wuhan National Laboratory for Optoelectronics, China, working on the implementation of TD-LTE systems as an R\&D engineer.
From 2016 to 2017, he has been a DSP engineer in Sequans Communications - an IoT chipmaker based in Paris, France. From 2017 to 2019, he has been a senior research engineer working on 5G standardization in Shanghai Huawei Technologies Co., Ltd., where he made substantial contributions to 5G standards, particularly the 5G codebooks. Since May 2019, he has joined the School of Electronic Information and Communications at Huazhong University of Science and Technology as a full professor.
His current research interests include 5G and 6G networks, signal processing, machine learning, and massive MIMO systems. H. Yin was the national champion of 2021 High Potential Innovation Prize awarded by Chinese Academy of Engineering, a winner of 2020 Academic Advances of HUST, and a recipient of the 2015 Chinese Government Award for Outstanding Self-financed Students Abroad.
\end{IEEEbiography}
\vfill

\begin{IEEEbiography}[{\includegraphics[width=1in,height=1.25in,clip,keepaspectratio]{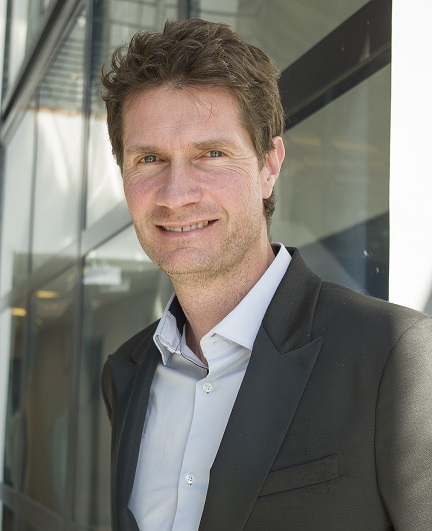}}]
{David Gesbert} (IEEE Fellow) is Professor and Director of EURECOM. He obtained the Ph.D. degree from Ecole Nationale Superieure des Telecommunications, France, in 1997. From 1997 to 1999 he has been with the Information Systems Laboratory, Stanford University. He was then a founding engineer of Iospan Wireless Inc, a Stanford spin off pioneering MIMO-OFDM (now Intel). Before joining EURECOM in 2004, he has been with the Department of Informatics, University of Oslo as an adjunct professor. D. Gesbert has published about 300 papers and 25 patents, some of them winning 2019 ICC Best Paper Award, 2015 IEEE Best Tutorial Paper Award (Communications Society), 2012 SPS Signal Processing Magazine Best Paper Award, 2004 IEEE Best Tutorial Paper Award (Communications Society), 2005 Young Author Best Paper Award for Signal Proc. Society journals, and paper awards at conferences 2011 IEEE SPAWC, 2004 ACM MSWiM. He has been a Technical Program Co-chair for ICC2017. He was named a Thomson-Reuters Highly Cited Researchers in Computer Science. Since 2015, he holds the ERC Advanced grant "PERFUME" on the topic of smart device Communications in future wireless networks. He is a Board member for the OpenAirInterface (OAI) Software Alliance. Since early 2019, he heads the Huawei-funded Chair on Advanced Wireless Systems Towards 6G Networks. He sits on the Advisory Board of HUAWEI European Research Institute.

\end{IEEEbiography}



\end{document}